\documentclass[12pt,authoryear]{elsarticle}

\usepackage[authoryear]{natbib}
\usepackage{appendix}

\usepackage{amsfonts}
\usepackage{amsmath}
\usepackage{amssymb}
\usepackage{natbib}
\usepackage{graphicx}
\usepackage{color}
\usepackage{rotating}
\usepackage{xcolor}
\usepackage{xr}
\usepackage{amsthm}
\usepackage{multirow}
\usepackage{booktabs}
\usepackage{cleveref}
\usepackage{float}
\usepackage{enumitem}

\newtheorem{cor}{Corollary}

\newtheorem{remark}{Remark}
\newtheorem{lemma}{Lemma}
\newtheorem{theorem}{Theorem}
\makeatletter
\def\ps@pprintTitle{%
	\let\@oddhead\@empty
	\let\@evenhead\@empty
	\let\@oddfoot\@empty
	\let\@evenfoot\@oddfoot
}
\makeatother

\newcommand{\fhat}[1]{\widehat{\mkern2mu#1\mkern2mu}}

\begin{document}

	\def\spacingset#1{\renewcommand{\baselinestretch}%
		{#1}\small\normalsize} \spacingset{1}


		\title{\bf Robust estimation in generalized linear models based on the normal quantiles of the probability integral transformation}
		\author{Marina Valdora\thanks{
				The authors gratefully acknowledge the support of Grants 20020170100330BA and 20020220200037BA
					from the University of Buenos Aires}\hspace{.2cm}\\
			Instituto de Cálculo,  University of Buenos  Aires and CONICET\\
			and \\
			Víctor Yohai \\
			Instituto de Cálculo and Department of Mathematics, FCEN,\\ University of Buenos Aires}
		\maketitle
	
	\bigskip
	\begin{abstract}
		A new approach to robust estimation in generalized linear models is introduced. The idea of the method is to first transform the responses applying the composition of the normal quantile function and the probability integral transformation. Then, using that the transformed responses should follow a standard normal distribution, find the values of the parameters that minimize a robust measure of their size. In practice an approximation of this transformation is used.
		The proposed estimators are studied theoretically for distributions that depend on a single parameter and through simulations and examples for the particular cases of Poisson and logistic regression. 
	\end{abstract}
	
	\noindent%
	{\it Keywords:} Robust statistics; Randomized quantile residuals; Robust Poisson regression; Robust logistic regression
	\vfill
	
	\newpage
	\spacingset{1.8} 

	\section{Introduction}\label{sec:intro}
	

	Robust estimation for generalized linear models  has been studied extensively since the 1980s, soon after \cite{nelder1972generalized} defined generalized linear models and introduced a method to compute maximum likelihood estimators. Among many important contributions to this topic we can cite \cite{kunsch1989conditionally}, \cite{bianco1996robust}, \cite{cantoni2001robust},  \cite{croux2003implementing}, \cite{bergesio2011projection}, \cite{valdora2014robust}, \cite{agostinelli2019initial}, \cite{aeberhard2014robust} and \cite{marazzi2019robust} . 
	
	In this paper we introduce a new class of estimators for GLMs based on the normal quantiles of the probability integral transformation (NQPIT). The idea of these estimators is, for each possible value of the parameters, to transform the responses using the PIT and then to apply the standard normal quantile function, following the ideas introduced in \cite{dunn1996randomized}. When the parameters used for the PIT coincide with the true values, the transformed responses will follow a standard normal distribution. After this transformation, one can use a loss function that has been proved to work for normal responses and search for values of the parameters that minimize this loss function.  In practice, a correction is needed to achieve Fisher-consistency. Moreover, in the case of discrete distributions an approximation of this transformation is used to avoid randomness.
	In this introductory section we explain this idea in more detail.   
	
	Let us first consider independent and identically distributed (i.i.d.) random variables $y_1, \dots, y_n$ from a distribution belonging to a parametric family  $\{ F(., \theta), \theta \in \Theta \}$.
	
	Consider the PIT:
	\begin{equation}\label{eq:PIT}
		t_P(y, u, \theta) = F(y,\theta) - u p(y, \theta), 
	\end{equation}
	where $u \sim U(0,1)$ is independent of $y$ and $p(k, \theta) = \mathbb P_\theta(y=k)$. Note that for continuous distributions $t_P(y, u, \theta)=F(y, \theta)$.

	It is known that   $t_P(y, u, \theta)$ follows a uniform distribution on the interval $[0,1]$ and that $t_N(y, u, \theta)=\Phi^{-1}\left(t_P(y, u, \mu)\right) \sim \Phi$, where $\Phi$ is the cumulative distribution function of a standard normal distribution.
	This fact has motivated the definition of randomized quantile (RQ) residuals; see \cite{dunn1996randomized},   and their implementation in the Dharma R package; see \cite{hartig2024dharma}. It has also motivated a class of estimators, introduced in \cite{fegyverneki2003robust} and also studied in
	\cite{valdora2020m}, where they were called MI-estimators.  MI-estimators are  defined as the solution to
	\begin{equation*}
		\sum_{i=1}^n  F(y,\theta) - \frac{1}{2} p(y, \theta) - \frac{1}{2} =0.
	\end{equation*}
	Now assume that $(y, \mathbf{x})$ follows a GLM with parameter $\boldsymbol{\beta}$ and link function $g$. Let $\theta=g^{-1}\left(\mathbf x^\top \boldsymbol{\beta}\right)$ and assume that the distribution of $y | \mathbf{x}$ belongs to a parametric family $\{ F(., \theta), \theta \in \Theta \}$, where $\Theta$ is an interval in $\mathbb R$ of the form $(\theta_1,  \theta_2).$  Suppose $\fhat{\boldsymbol{\beta}}$ is an estimator of $\boldsymbol{\beta}$ and let 
	 $\fhat{\theta}_{\mathbf x}=g^{-1}\left(\mathbf x^\top \fhat{\boldsymbol{\beta}}\right)$. 
	For a random sample $(y_1, \mathbf{x}_1), \dots, (y_n, \mathbf{x}_n)$ with the same distribution as $(y, \mathbf{x})$,
	RQ residuals are defined as $q_i=t_N(y_i, u_i, \fhat{\theta}_i)$, where $u_i$ are random variables uniformly distributed on the interval $[0,1]$ independent of $y_1, \dots, y_n$,  and $\fhat{\theta}_i=\fhat{\theta}_{\mathbf x_i}$.
	Since $t_N(y_i, u_i, {\theta}_{\mathbf x_i})$  follows a standard normal distribution for $i = 1, \dots, n$, it seems natural to define a least squares estimator based on the NQPIT by minimizing the sum of the squared RQ residuals and 
	to attempt to robustify this estimator by replacing the quadratic function with a bounded $\rho$-function. 
	Two technical problems arise if we define the estimator as the value that minimizes
	\begin{equation}\label{eq:est1withoutcorrection}
		\sum_{i=1}^n  \left(t_N\left(y_i, u_i ,g^{-1}\left(\mathbf x_i^\top \boldsymbol{\beta}\right)\right)\right) ^2.
	\end{equation}
	The first is that, while the classical least squares estimator is Fisher-consistent for linear models, the estimator defined as the minimizer of \eqref{eq:est1withoutcorrection} is not. 
	An estimator defined as the minimizer of a loss function is said to be Fisher-consistent if, under the true distribution, the expected loss function is minimized at the true value of the parameter; see \cite{maronna2019robust}.  
	To make the proposed estimator Fisher-consistent, a correction function needs to be introduced. This will be done in Section \ref{sec:LSNQPIT}.
	
	The second problem is that, in the case of discrete responses, if we defined the estimator simply as the minimizer of \eqref{eq:est1withoutcorrection}, it would be dependent on the random variables $u_i$. To avoid this dependence, we replace the random variables $u_i$ by $\mathbb{E}(u_i)=0.5$ and consider the estimator defined as the value that minimizes
	\begin{eqnarray}\label{eq:esteqwithoutcorrection}
		\sum_{i=1}^n  \left(  t\left(y_i,g^{-1}\left(\mathbf x_i^\top \boldsymbol{\beta}\right)\right)\right)^2,  
	\end{eqnarray}
	where
	\begin{equation}\label{eq:deft}
		t(y, \gamma)=\Phi^{-1} \left( F(y,\gamma) - 0.5 \, p(y, \gamma)\right). 
	\end{equation}Our experiments showed that this simplification increases the efficiency and robustness of the estimator, even though it results in a transformation that is no longer normal and has variance only approximately equal to 1. Of course, this second problem does not arise in the case of continuous distributions, since the PIT does not depend on $u$.
	
	The rest of the paper is organized as follows: in Section \ref{sec:LSNQPIT} we introduce least squares estimators based on the normal quantiles of the probability integral transformation, which are not robust but help to explain the idea of our robust proposal. In Section \ref{sec:mpit} we introduce the proposed robust estimators for GLMs based on the NQPIT, which we call M-estimators based on the normal quantiles of the probability integral transformation (MNQPIT-estimators). In Section \ref{sec:asprop} we state asymptotic properties of MNQPIT- estimators: their consistency, asymptotic distribution and asymptotic breakdown point (ABP). In Section \ref{sec:implementation} we describe the implementation of the proposed estimators for Poisson and binomial GLMs. In Section \ref{sec:montacarlo} we describe the results of a Monte Carlo Study in which we compare the performance of our proposed estimators to that of classical ones and of some of the robust estimators existing in the literature, in the case of Poisson and binomial GLMs. In Section \ref{sec:realdata} we present two real data examples.
	Supplementary material containing the proofs of the lemmas and theorems stated in this article and the R code used in the simulations and examples can be found on line.
	
	\section{Least squares estimators based on the normal quantiles of the probability integral transformation}\label{sec:LSNQPIT}
	The estimator defined as the minimizer of \eqref{eq:esteqwithoutcorrection}  is not Fisher consistent because 
	$\operatorname{argmin}_{\gamma\in\Theta} \mathbb{E}_\theta \left(  t(y,  \gamma)^2 \right) \neq \theta.$
	Define the following correction function for quadratic loss  $m(\theta)=\operatorname{argmin}_{\gamma\in\Theta} \mathbb{E}_\theta \left(  t(y,  \gamma)^2\right)$.
	Then, by definition, $$\operatorname{argmin}_{\gamma\in\Theta} \mathbb{E}_\theta \left(  t(y, m(\gamma))^2\right)=\theta.$$
	
	Least squares estimators based on the normal quantiles of the probability integral transformation (LSNQPIT-estimators) are then defined as \begin{equation}\label{eq:estwithcorrection}
		\fhat{\boldsymbol{\beta}}_{LSNQPIT}= \operatorname{argmin}_{\boldsymbol{\beta}}	\sum_{i=1}^n t\left(y_i, m\left(g^{-1}\left(\mathbf x_i^\top \boldsymbol{\beta}\right)\right)\right)^2.
	\end{equation}
	By the definition of $m$, LSNQPIT-estimators are Fisher consistent; see Lemma 2 in Appendix A. 
	It is also easy to see that under the normal model with known variance, $m$ is the identity function and the LSNQPIT-estimator coincides with the LS estimator.
	
	Our experiments showed that the estimators defined in this way are highly efficient, giving results that almost coincide with the maximum likelihood estimates in several simulation settings. We do not report details of these simulations because these estimators are not robust. 

	\section{M-estimators based on the normal quantiles of the probability integral transformation }\label{sec:mpit}
	
	We now aim at robustifying LSNQPIT-estimators. Recall that, after the simplification and the correction, the transformed responses $t\left(y, m\left(g^{-1}\left(\mathbf x_i^\top \boldsymbol{\beta}\right)\right)\right)$ do not follow a normal distribution. However, their location and scale should be similar to the standard normal's and therefore, the methods for robust estimators for normal models will guide the construction of our robust estimator. Moreover, using the fact that the variance is approximately 1, we do not need to divide by a scale estimator. Let $\rho$ be a bounded $\rho$-function; see 
	Assumptions \ref{as:rhoeven}-\ref{as:rhocont}. 
	Let \begin{equation} \label{eq:mdef}
		m(\theta)=\operatorname{argmin}_{\gamma\in\Theta} \mathbb{E}_\theta \left( \rho\left(t(y, \gamma)\right)\right).
	\end{equation}
	and \begin{equation*}
		L_n(\boldsymbol{\beta})=\sum_{i=1}^n \rho\left(t\left(y_i,m\left(g^{-1}\left(\mathbf x_i^\top \boldsymbol{\beta}\right)\right)\right)\right).
	\end{equation*}
	Then MNQPIT-estimators as defined as \begin{equation}\label{eq:est}
		\fhat{\boldsymbol{\beta}}_{MNQPIT}= \operatorname{argmin}_{\boldsymbol{\beta}}		L_n(\boldsymbol{\beta}).\end{equation}
	
	By the definition of $m$, MNQPIT-estimators are Fisher consistent; see Lemma 2 
	in Appendix A. 
	
	In some models, such as logistic regression, to obtain the required degree of  robustness, it is necessary to down-weight high-leverage observations; see \cite{croux2003implementing}. For this reason we define a wider class of estimators, called weighted M-estimators based on the normal quantiles of the probability integral transformation (WMNQPIT-estimators), as the value of $\boldsymbol{\beta }$ that minimizes
	\begin{equation}\label{eq:wmpitest}
		L_n(\boldsymbol{\beta})=\sum_{i=1}^n \rho\left(t\left(y_i,m\left(g^{-1}\left(\mathbf x_i^\top \boldsymbol{\beta}\right)\right)\right)\right) w(\mathbf{x}, \fhat{\boldsymbol\mu}_n , \fhat{\boldsymbol\Sigma}_n),
	\end{equation}
	where $w(\mathbf{x}, \boldsymbol{\mu}, \mathbf{\Sigma})$ is a function of the Mahalanobis distance, that is
	\begin{equation}\label{eq:weight}
		w(\mathbf{x}, \boldsymbol{\mu}, \boldsymbol{\Sigma})=\omega \left(\left((\mathbf{x}-\boldsymbol{\mu})^{\prime} \mathbf{\Sigma}^{-1}(\mathbf{x}-\boldsymbol{\mu})\right)^{1 / 2}\right),
	\end{equation}
	$\fhat{\boldsymbol{\mu}}_n$ and $\fhat{\boldsymbol{\Sigma}}_n$ are robust estimators of location and scatter matrix of $\mathbf{x}$ based on $\mathbf{x}_1, \ldots, \mathbf{x}_n$ and $\omega$ is a non-negative non-increasing function. The purpose of the weighting function $w
	$ 
	is to penalize high leverage observations. We will use consistent estimators $\fhat{\boldsymbol{\mu}}_n$ and $\fhat{\boldsymbol{\Sigma}}_n$ so that $\fhat{\boldsymbol{\mu}}_n \rightarrow \boldsymbol{\mu}_0$ almost surely (a.s.) and $\fhat{\boldsymbol{\Sigma}}_n \rightarrow \boldsymbol{\Sigma}_0$ a.s., where $\boldsymbol{\mu}_0$ and $\boldsymbol{\Sigma}_0$ are parameters of location and scatter of $\mathbf{x}$.  
	
	It is worth mentioning that, since $\rho$ is bounded, in many cases it is not necessary to down-weight high levererage observations, that is, we can achieve very high robustness taking $\omega \equiv 1$, as we will see in the simulation study. Moreover, as mentioned in \cite{valdora2014robust}, in some cases the use of weights may decrease the robustness of the estimator. This occurs when there are good high leverage observations, that is to say, observations with high leverage but with response $y$ following the GLM of the majority of the data. In these cases the weight function may practically eliminate good observations and therefore it may increase the influence of the outliers with low leverage that have larger weights. For this reason, we limit the use of weights to the case of logistic regression, in which their need has been established by \cite{croux2003implementing}, among others.
	

	\section{Asymptotic properties of MNQPIT-estimators}\label{sec:asprop}
	\subsection{Consistency}
	Assume that $(y, \mathbf x)$ follows a GLM with parameter $\boldsymbol{\beta}$, link function $g$ and distribution function $F$. Assume $F$ belongs to a parametric family $\{ F(., \theta), \theta \in \Theta \}$, where $\Theta$ is an interval $(\theta_1,  \theta_2)\subset  \mathbb R$ {and that $y$ is either continuous with probability density function $f(., \theta)$ or discrete with probability mass function $p(., \theta)$}. Denote $m_1=\lim_{\theta \rightarrow \theta_1} m(\theta)$, $m_2=\lim_{\theta \rightarrow \theta_2} m(\theta)$ and $S=\left\{\mathbf{t} \in \mathbb{R}^p:\|\mathbf{t}\|=1\right\}$.  Note that, throughout this paper, for both continuous and discrete distributions, we denote  $p(k, \theta)=\mathbb{P}_\theta(y=k)$, which equals zero if $y$ is continuous.

	To prove the consistency of MNQPIT-estimators we need the following assumptions. \begin{enumerate}[label=\textbf{A\arabic*},resume=AList]
		\item\label{as:Fcont} The cumulative distribution function $F(y,\theta)$ is continuous as a function of $\theta$. 
		\item\label{as:Fmonotentheta} If $\theta_1<\theta_2$, then $F(y,\theta_1) \geq F(y, \theta_2)$ for all $y$, and there exists $y$ such that $F(y, \theta_1)>F(y, \theta_2)$.
		\item \label{as:limF} {The support of $y$ and the set $\{k\in \mathbb R / F(k, \theta)<1) \}$ are independent of $\theta$. Moreover $\lim_{\theta\rightarrow\theta_1} F(k,\theta)=1$  and $\lim_{\theta\rightarrow\theta_2} F(k,\theta)=0$ for all $k$ in the support of $y$ such that $F(k, \theta) < 1$.} 
		\item \label{as:linkfunction} The link function $g:\Theta \to \mathbb R$ is continuous and strictly increasing and verifies that $\lim_{\theta\rightarrow\theta_1}g(\theta)=-\infty$ and $\lim_{\theta\rightarrow\theta_2}g(\theta)=+\infty$. 
		\item\label{as:mwelldefined} The function $m$ is univocally defined and strictly increasing. 
		\item\label{as:existeepsilon0} There exist $c, \epsilon_0>0$ such that $$\rho(c) + \frac{1}{12\left(\Phi(c) - 0.5\right)^2} < 1-\epsilon_0$$
		
		\item \label{as:rhoeven} $\rho(u) \geq 0, \rho(0)=0 \text { and } \rho(u)=\rho(-u)$
		\item \label{as:limrho} $\lim _{u \rightarrow \infty} \rho(u)=1$. 
		\item \label{as:rhoincreasing} $0 \leq u<v$ implies $\rho(u) \leq \rho(v)$.
		$0 \leq u<v$ and $\rho(u)<1$ implies $\rho(u)<\rho(v)$.
		\item \label{as:rhocont} $\rho$ is continuous
		\item \label{as:limmusigma} There exist $\boldsymbol{\mu}_0 \in \mathbb{R}^\rho$ and a positive definite matrix $\boldsymbol{\Sigma}_0$ such that $\boldsymbol{\mu}_n \rightarrow \boldsymbol{\mu}_0$ a.s. and $\boldsymbol{\Sigma}_n \rightarrow \boldsymbol{\Sigma}_0$ a.s.
		\item \label{as:limw} The weight function $\omega$ is continuous, bounded and non-increasing and $\sup \omega=1$.
		\item \label{as:prob} The covariates $\mathbf x$ and weight function $w$ verify that
		$$\inf _{\mathbf{t} \in S} P\left(\left\{\mathbf{x}^{\top} \mathbf{t} \neq \mathbf{0}\right\} \cap\left\{{w}\left(\mathbf{x}, \boldsymbol{\mu}_0, \Sigma_0\right)>0\right\}\right)>0
		$$
	\end{enumerate}
	
	\begin{remark}
		{Since replacing $\rho$ by $k\rho$ yields the same estimator for all $k>0$}, Assumption \ref{as:limrho} may
		be replaced by the assumption that $\rho$ is bounded. A similar remark applies to Assumption \ref{as:limw}.
		The choice of setting both bounds equal to 1 is made to simplify the proofs.
	\end{remark}
	\begin{remark}
		Assumptions \ref{as:Fcont} to \ref{as:limF} hold for Poisson, Bernoulli and exponential distributions. {They also hold for the normal distribution with known variance. Assumption \ref{as:linkfunction} holds, for example, for the log link in Poisson regression and for logistic regression}.
	\end{remark}
	

	\begin{remark}
		{Assumption \ref{as:mwelldefined} holds for Tukey's bisquare $\rho-$function and the specific $\rho-$function and the calibration constants used in this paper for Poisson and logistic distributions; see Section \ref{sec:implementation}. This has been verified in several numerical experiments; see Figure \ref{fig:mbinom}. However we still lack a formal proof of this fact.}
	\end{remark}
	\begin{figure}
		\centering
		\begin{tabular}{cc}
			\includegraphics[width=0.45\linewidth]{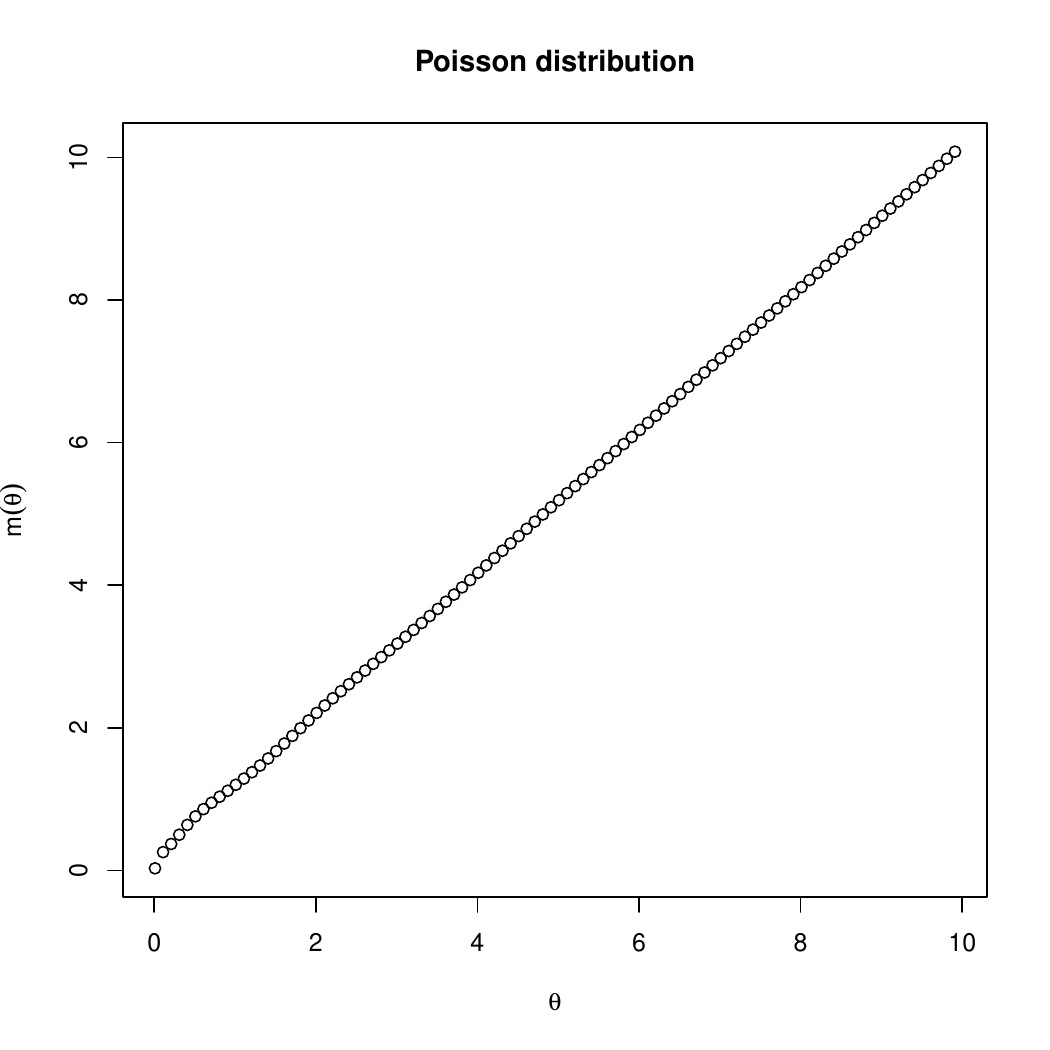}&
			\includegraphics[width=0.45\linewidth]{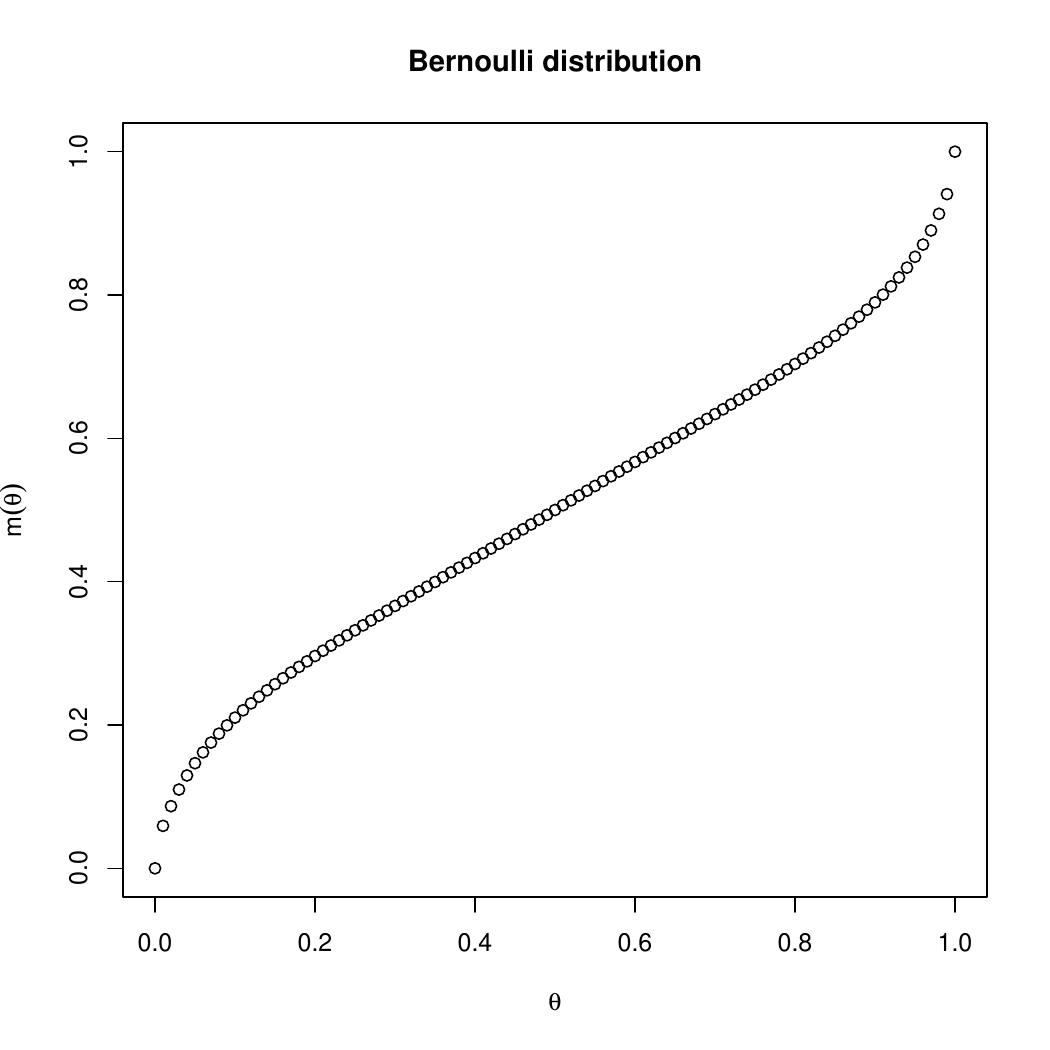}
		\end{tabular}
		\caption{$m$ function for Poisson distribution (left) and Bernoulli distribution (right) computed with the $\rho$ function and calibration constant described in Section \ref{sec:implementation}.}
		\label{fig:mbinom}
	\end{figure}

	\begin{remark}{
			Assumption \ref{as:existeepsilon0} 
			is true for several $\rho$ functions and distributions. In particular, if we take the $\rho$ function defined in  Section \ref{sec:implementation}  calibrated for asymptotic efficiency of {$0.95$, as we did in our simulations and examples, we can take  $c=0.3$ and $\epsilon_0=0.84$, so $1-\epsilon_0=0.16$. This assumption is used in Lemma 5 in Appendix A, 
				where we prove  that  $\mathbb{E}_\theta \left( \rho\left(t(y, \theta)\right)\right)<1-\epsilon_0$. } 
		}
	\end{remark}
	\begin{remark}{
			Assumptions \ref{as:rhoeven} to \ref{as:rhocont} are standard properties for $\rho$-functions in robust statistics, while Assumptions \ref{as:limmusigma} to \ref{as:prob} are very general assumptions on the weight functions in WMNQPIT-estimators.}
	\end{remark}

	For any set $A\subset \mathbb R^p$ and $\mathbf z \in \mathbb R^p$, we denote by $I(\mathbf z\in A)$ the function that indicates if $\mathbf z\in A$, that is to say, that equals $1$ if $\mathbf z\in A$ and $0$ otherwise.
	For $z, l \in \mathbb R$, we denote by $I(z<l)$ the function that indicates if $z<l$, that is to say, that equals $1$ if $z<l$ and $0$ otherwise. We define $I(z\geq l)$ in an analogous manner.	
	The following theorem states the consistency of WMNQPIT-estimators.
	
	\begin{theorem}[Consistency]\label{teo:consistency}
		Let $\left(y_i, \mathbf{x}_i\right), i \in \mathbb N$, be a sequence of i.i.d. random vectors following a GLM with parameter $\boldsymbol{\beta}_0$, link function $g$ and distribution function $F$. 
		Let $\fhat{\boldsymbol{\beta}}_n$ be the WMNQPIT-estimator of $\boldsymbol{\beta}_0$  defined by \eqref{eq:wmpitest} and assume \ref{as:Fcont} to \ref{as:prob} hold. Let 
		\begin{align}\label{eq:tau}\small
			&\tau= \inf_{\left\| \mathbf{t} \right\| = 1} \lim_{\gamma\rightarrow\infty} \mathbb{E}_{\boldsymbol\beta_0}\left[ \left(\rho\left(t\left(y, m \left(g^{-1}\left(\mathbf{x}^{\top} \gamma \mathbf{t}  \right)\right)\right)\right) - \right. \right. \\ \nonumber & \quad\quad\quad\quad\quad\quad\quad\quad\quad\quad \left.\left.\rho\left(t\left(y, m \left(g^{-1}\left(\mathbf{x}^{\top} \boldsymbol\beta_0  \right)\right)\right)\right)\right) w\left(\mathbf{x}, \boldsymbol{\mu}_0, \boldsymbol{\Sigma}_0\right)\right]
		\end{align}
		Then, $\tau>0$ and, if $\mathbb P\left(\mathbf{t}^{\prime} \mathbf{x}=0\right)<\tau$ for all $\mathbf{t}$ with $||\mathbf t||=1$, then $\fhat{\boldsymbol{\beta}}_n \rightarrow \boldsymbol{\beta}_0$ a.s.. 
	\end{theorem}
	Theorem \ref{teo:consistency} is proved in Appendix A in the supplementary material, as a consequence of Theorem 4 in \cite{valdora2014robust}. 
	\subsection{Asymptotic normality}
	To establish asymptotic normality we will need the following additional assumptions:
	\begin{enumerate}[label=\textbf{A\arabic*}, resume=AList]
		\item\label{as:gC2} The link function $g$ is twice continuously differentiable.
		\item \label{as:Fdifer} $F(y,\theta)$ has three continuous and bounded derivatives as a function of $\theta$. \item \label{as:rhodifer} $\rho$ has three continuous and bounded derivatives. We write $\psi=\rho^{\prime}$.
		\item \label{as:psiprimanotzero} $\mathbb E_\theta\left(\psi^{\prime}\left(t(y,m(\theta))\right)t^{\prime}(y,m(\theta))+\psi \left(t(y,m(\theta))\right)t^{\prime\prime}(y,m(\theta))\right) \neq 0$ for all $\theta$, where $t^\prime(y, \theta)$ denotes the partial derivative of $t(y, \theta)$ with respect to $\theta$.
		
	\end{enumerate}
	Let $\boldsymbol{\Psi}=\left(\Psi_1, \ldots \Psi_p\right): \mathbb{R} \times \mathbb{R}^p \times \mathbb{R}^p \rightarrow \mathbb{R}^p$ be defined by
	$$\small
	\begin{aligned}
		&\Psi _j(y, \mathbf{x}, \boldsymbol{\beta}, \boldsymbol{\mu}, \mathbf{\Sigma})=w(\mathbf{x}, \boldsymbol{\mu}, \mathbf{\Sigma}) \frac{\partial}{\partial \boldsymbol{\beta}_j} \rho\left(t\left(y,m\left(g^{-1}\left( \mathbf{x} ^\top \boldsymbol{\beta}\right)\right)\right)\right)\\
		& \quad=w(\mathbf{x}, \boldsymbol{\mu}, \mathbf{\Sigma}) \psi\left(t\left(y, m\left( g^{-1}\left( \mathbf{x} ^\top \boldsymbol{\beta}\right)\right)\right)\right)t^{\prime}\left(y,m\left(g^{-1}\left( \mathbf{x} ^\top \boldsymbol{\beta}\right)\right)\right) m^{\prime}\left(g^{-1}\left( \mathbf{x} ^\top \boldsymbol{\beta}\right)\right) g^{-1\prime}\left( \mathbf{x} ^\top \boldsymbol{\beta}\right) x_j,
	\end{aligned}
	$$
	
	Denote by $\mathbf{J}_{\boldsymbol\Psi}(y, \mathbf{x}, \boldsymbol{\beta}, \boldsymbol{\mu}, \boldsymbol{\Sigma})$ the Jacobian matrix of $\boldsymbol{\Psi}$ with respect to $\boldsymbol{\beta}$, that is to say, the matrix with entries 
	$$
	J_{\boldsymbol\Psi}^{j, k}(y, \mathbf{x}, \boldsymbol{\beta}, \boldsymbol{\mu}, \mathbf{\Sigma})=\frac{\partial}{\partial \boldsymbol{\beta}_k} {\Psi}_j(y, \mathbf{x}, \boldsymbol{\beta}, \boldsymbol{\mu}, \mathbf{\Sigma}), \quad 1 \leq j, k \leq p.
	$$

	Note that Assumptions \ref{as:Fdifer}, \ref{as:rhodifer} and the differentiability of $m$, proved in Lemma 8 
	in Appendix A, imply that $\Psi$ and $\mathbf{J}_\psi$ are well defined. 
	\begin{enumerate}[label=\textbf{A\arabic*}, resume=AList]
		\item \label{as:existepsilon} There exists $\varepsilon>0$ such that $\mathbb E_{\beta_0}\left(\sup _{\left\|\beta-\beta_0\right\| \leq \varepsilon} \left| J_{\boldsymbol\Psi}^{j, k}\left(y, \mathbf{x}, \boldsymbol{\beta}, \boldsymbol{\mu}_0, \boldsymbol{\Sigma}_0\right) \right|\right)<\infty$, for all $1 \leq j, k \leq p$, where $\left\| \,\, \|\right.$ denotes the $l_2$ norm, and $ \mathbb E_{\boldsymbol{\beta}_0}\left(\mathbf{J}_{\boldsymbol\Psi}\left(y, \mathbf{x}, \boldsymbol{\beta}_0, \boldsymbol{\mu}_0, \boldsymbol{\Sigma}_0\right)\right)$ is non-singular.
	\end{enumerate}
	\begin{remark}
		{Assumptions \ref{as:gC2} to \ref{as:psiprimanotzero} are needed to have the differentiability of $\boldsymbol{\Psi}$, while Assumption \ref{as:existepsilon} is needed to obtain the expression of the asymptotic covariance matrix. Both assumptions have been verified empirically in the case of Poisson and logistic regression models and the $\rho$-function defined in Section \ref{sec:implementation}.}\end{remark} 
	%
	\begin{theorem}[Asymptotic normality]\label{teo:asnorm} Let $\left(y_i, \mathbf{x}_i\right), i \in \mathbb N$, be a sequence of i.i.d. random vectors following a GLM with parameter $\boldsymbol{\beta}_0$, link function $g$ and distribution function $F$.  Let $\fhat{\boldsymbol{\beta}}_n$ be the WMNQPIT-estimator of $\boldsymbol{\beta}_0$  defined by \eqref{eq:wmpitest} and assume \ref{as:Fcont} to \ref{as:existepsilon} hold.
		Assume also that  
		$P\left(\mathbf{t}^{\prime} \mathbf{x}=0\right)<\tau / M$ for all $\mathbf{t} \in S$, where $M=\sup \rho$ and $\tau$ is given by \eqref{eq:tau}. Then
		$$\sqrt{n}\left(\fhat{\boldsymbol{\beta}}_n-\boldsymbol{\beta}_0\right) \xrightarrow{\mathcal{D}} \mathcal{N}\left(0, \mathbf{B}^{-1} \mathbf{A B ^ { \top -1}}\right),$$
		where $\mathcal{N}_p(\boldsymbol{\mu}, \mathbf{\Sigma})$ denotes the p-dimensional multivariate normal distribution with mean $\boldsymbol{\mu}$ and covariance matrix $ \boldsymbol{\Sigma}$,
		$\mathbf{B}=E_{\boldsymbol{\beta}_0}\left(\mathbf{J}_{\boldsymbol{\Psi}} \left(y, \mathbf{x}, \boldsymbol{\beta}_0, \boldsymbol{\mu}_0, \boldsymbol{\Sigma}_0\right)\right)
		$, and
		$
		\mathbf{A}=E_{\boldsymbol{\beta}_0}\left(\boldsymbol{\Psi}\left(y, \mathbf{x}, \boldsymbol{\beta}_0, \boldsymbol{\mu}_0, \boldsymbol{\Sigma}_0\right) \boldsymbol{\Psi}\left(y, \mathbf{x}, \boldsymbol{\beta}_0, \boldsymbol{\mu}_0, \boldsymbol{\Sigma}_0\right)^{\top}\right)$. 
	\end{theorem}
	\subsection{Asymptotic variance}

	Denote $\mu_\mathbf x=g^{-1}\left(\boldsymbol{\beta}^{\top}_0 \mathbf{x}\right)$, $\mu^{\prime}_\mathbf x=g^{-1\prime}\left(\boldsymbol{\beta}^{\top}_0 \mathbf{x}\right)$,   $m_\mathbf x = m(\mu_\mathbf x)$, $m^\prime_\mathbf x = m'(\mu_\mathbf x)$, $m^{\prime\prime}_\mathbf x = m^{\prime\prime}(\mu_\mathbf x)$  and $\mu^{\prime\prime}_\mathbf x=g^{-1\prime\prime} \left(\boldsymbol{\beta}^{\top}_0 \mathbf{x}\right)$. Then
	$$\boldsymbol{\Psi}\left(y, \mathbf{x}, \boldsymbol{\beta}_0, \boldsymbol{\mu}_0, \boldsymbol{\Sigma}_0\right) \boldsymbol{\Psi}\left(y, \mathbf{x}, \boldsymbol{\beta}_0, \boldsymbol{\mu}_0, \boldsymbol{\Sigma}_0\right)^{\top} = \mathbf x\left(\psi\left(t\left(y, m_\mathbf x\right)\right)t^{\prime}\left(y, m_\mathbf x\right) m^{\prime}_\mathbf x \mu^{\prime}_\mathbf x\right)^2 \mathbf x^\top$$  and $$\mathbf{J}_{\boldsymbol\Psi}(y, \mathbf{x}, \boldsymbol{\beta}, \boldsymbol{\mu}, \boldsymbol{\Sigma}) = \mathbf x \nu(\mathbf x, y, \boldsymbol{\beta}) w(\mathbf{x}, \boldsymbol{\mu}, \mathbf{\Sigma}))  \mathbf x^\top,$$ where
	$\nu(\mathbf x, y, \boldsymbol{\beta}) =
	\psi^{\prime}\left(t\left(y, m_\mathbf x\right)\right)
	\left(t^{\prime}\left(y,m_\mathbf x\right) m^{\prime}_\mathbf x \mu^{\prime}_\mathbf x \right)^2 +
	\psi\left(t\left(y, m_\mathbf x\right)\right)t^{\prime \prime}\left(y,m_\mathbf x \right) \left(m^{\prime}_\mathbf x  \mu^{\prime}_\mathbf x \right)^2+ 
	\psi\left(t\left(y, m_\mathbf x\right)\right)t^{\prime}\left(y,m_\mathbf x\right) m^{\prime \prime}_\mathbf x \mu^{\prime 2}_\mathbf x +
	\psi\left(t\left(y, m_\mathbf x\right)\right)t^{\prime}\left(y, m_\mathbf x\right) m^{\prime}_\mathbf x \mu^{\prime\prime}_\mathbf x.$
	
	Therefore 
	$$\mathbf A = \mathbb E\left( \mathbf x  \, \mathbb E_{\boldsymbol{\beta}_0} \left(\left(\psi\left(t\left(y, m_\mathbf x\right)\right)t^{\prime}\left(y, m_\mathbf x\right)^2\right)| \mathbf x \right) \left(m^{\prime}_\mathbf x \mu^{\prime}_\mathbf x\right)^2 \mathbf x^\top\right)$$
	and  $$ \mathbf B=\mathbb{E}\left(\mathbf x \, \mathbb{E}_{\boldsymbol{\beta}_0} \left(\nu(\mathbf x, y, \boldsymbol{\beta}_0)|\mathbf x\right) w(\mathbf{x}, \boldsymbol{\mu}, \mathbf{\Sigma})  \mathbf x^\top\right).$$
	We can estimate these matrices by replacing  $\boldsymbol{\beta}_0$ with $\fhat{\boldsymbol{\beta}}$ and the expectations with respect to $\mathbf x$ by averages. In this way we can obtain $\fhat{\mathbf{A}}$ and $\fhat{\mathbf{B}}$ and an estimate of the asymptotic covariance matrix as
	$\fhat{V} = \fhat{\mathbf{B}}^{-1} \fhat{\mathbf{A}} \fhat{\mathbf{B}} ^ { \top -1}$.
	
	\subsection{Asymptotic breakdown point}
	In this section, we compute the ABP of MNQPIT-estimators. The ABP is a measure of robustness of an
	estimator introduced by \cite{hampel1971general}. 
	Let $(y,\mathbf{x})$ be a random
	vector with values in n $\mathbb{R}\times$
	$\mathbb{R}^{p}$ that follows a GLM with parameter $\boldsymbol{\beta}$ and link function $g$. Let $H$ be the distribution of  $(y,\mathbf{x})$ and consider an estimator defined by a functional $\mathbf T$, that is to say, given a sample $(y_{1},\mathbf{x}_{1}%
	),...(y_{n},\mathbf{x}_{n}),$  $\fhat{\boldsymbol{\beta}}_n = \mathbf T( H_n)$, where $H_n$ is the empirical distribution of the sample. 
	For example, the  MNQPIT functional is
	\begin{equation}
		\mathbf T(H)=\operatorname{argmin}_{\boldsymbol{\beta}\in{\mathbb{R}}^{q}}\mathbb E_{H}\left(
		\rho\left(t\left(y,m\left(g^{-1}\left(\mathbf{x}^{\top}\boldsymbol{\beta}\right)\right)\right)\right)\right).\label{eq:functional}
	\end{equation}
	The ABP of the estimator $\fhat{\boldsymbol{\beta}}_n$ is defined as the ABP of the functional $\mathbf T$, as
	\[
	\varepsilon^{\ast}(\mathbf{T},H_{0})=\sup_{\varepsilon}\{\varepsilon
	\in(0,1):\sup_{H^{\ast}\in\mathcal{D}}\{||\mathbf{T}(1-\mathbf{\varepsilon
	})H_{0}+\mathbf{\varepsilon}H^{\ast})||_2\}<\infty\},
	\]
	where 
	$\mathcal{D}$ the set of all the distributions on $\mathbb{R}\times$
	$\mathbb{R}^{p}$ and $||$ $||_{2}$ denotes the $l_{2}$ norm. See \cite{maronna2019robust} for more details on this robustness measure.
	
	\begin{theorem}  Assume that $(y, \mathbf x)$ follows a GLM with parameter $\boldsymbol{\beta}_0$, link function $g$ and distribution function $F$ and let $\theta_{\mathbf x}=g^{-1}\left(\mathbf x^\top \boldsymbol{\beta}_0\right)$. Assume $F$ belongs to a parametric family $\{ F(., \theta), \theta \in \Theta \}$, where $\Theta$ is an interval in $\mathbb R$ of the form $(\theta_1,  \theta_2).$
		Let $H_0$ be the joint distribution of $(y,\mathbf{x})$ and assume
		$P_{H_{0}}(\mathbf{x}^\top\boldsymbol{\alpha}=0)=0$ for all
		$\boldsymbol\alpha\in\mathbb{R}^{p}$. 
		Assume \ref{as:Fcont} to \ref{as:limF},  \ref{as:mwelldefined}, \ref{as:rhoeven} and \ref{as:limrho}. 	Let $\phi_1^*(y )= \lim_{\theta\rightarrow \theta_1} \rho \left( t \left(y, m(\theta)\right)\right) $,  $\phi_2^*(y )= \lim_{\theta\rightarrow \theta_2} \rho\left( t \left(y, m(\theta)\right)\right)$ and
		\begin{equation}\label{eq:bp}
			\varepsilon_0\left(\boldsymbol{\beta }_0\right)=\frac{ \mathbb E_{H_0}\left( \min \left( \phi_1^*(y ),\phi_2^*(y ) \right)  \right)  -   \mathbb E_{H_0}\left(\rho  \left(t \left(y, m\left(\theta_{\mathbf{x}}  \right)\right)\right)\right)}{ 1 +  \mathbb E_{H_0}\left(  \min \left( \phi_1^*(y ),\phi_2^*(y ) \right)\right)   -E_{H_0}\left(\rho \left(t \left(y, m\left( \theta_{\mathbf x}\right)\right)\right)\right)}
		\end{equation}
		Then $\varepsilon_0\left(\boldsymbol{\beta }_0\right)$ is a lower bound for  the asymptotic  breakdown point of the MNQPIT-estimator. 	\label{teo:bp1}
	\end{theorem}
	The proof of Theorem \ref{teo:bp1} is given in Appendix A in the supplementary material.
	
	We now focus on the breakdown point of MNQPIT-estimators for GLMs with discrete response, such as Poisson and binomial.
	
	Consider the following additional assumptions on the distribution of $y$.
	\begin{enumerate}[label=\textbf{A\arabic*}, resume=AList]
		\item \label{as:ydiscreteN0} $y$ is a discrete random variable with rank $R \subset \mathbb N_0$.
		\item \label{as:limsp}  $\lim_{\theta\rightarrow \theta_1} p(y, \theta) = I_{\{0\}}(y)$. 
	\end{enumerate}
	
	\begin{theorem}
		Assume the hypothesis of Theorem \ref{teo:bp1} and Assumptions \ref{as:linkfunction}, \ref{as:rhoincreasing}, \ref{as:rhocont}, \ref{as:ydiscreteN0} and \ref{as:limsp} hold. 
		Let
		$$\varepsilon_0(\boldsymbol{\beta}_0)=\frac{1 - \mathbb P_{H_0}(y=0)-\mathbb E_{H_0}\left(\rho \left( t \left(y, m(\theta_{\mathbf{x}} ) \right)\right)\right)}{ 2 - \mathbb P_{H_0}(y=0)- \mathbb E_{H_0}\left(\rho \left(t \left(y, m(\theta_{\mathbf{x}})\right)\right)\right)},
		$$
		then $\varepsilon_0(\boldsymbol{\beta}_0)$ is a lower bound for  the asymptotic  breakdown point of the functional \eqref{eq:functional}  in $H_0$.	\label{teo:bp}
	\end{theorem}
	
	\begin{remark}	It is straightforward to verify that the binomial and the Poisson distributions verify Assumptions \ref{as:ydiscreteN0} and \ref{as:limsp} and therefore, under Assumptions \ref{as:linkfunction} to \ref{as:rhocont}, $\varepsilon_0$ given by Theorem \ref{teo:bp} is a lower bound for the ABP.
	\end{remark}
	
	\begin{remark} \label{obs:ABP} Note that our lower bound for the ABP depends on the true value of $\boldsymbol{\beta}_0$.
		If $\boldsymbol{\beta}_0$ is such that the probability that $y=0$ is large, then $\varepsilon_0(\boldsymbol{\beta}_0)$ will be near zero. 
		Note that in this case, the probability that $y=0$ is large and a small fraction of outliers equal to 0 can make the fraction of observed zeros larger than 0.5. Therefore the good non-null observations may be mistaken for outliers and in this case, the estimator will break down, giving an estimate that makes $\mathbb{P}(y=0)\approx 1$. This is seen in the example in the following section.
	\end{remark}	
	
	\subsubsection{The ABP for the case of Poisson and logistic regression with normal covariates}
	
	We computed our lower bound for the ABP numerically for the particular case in which the vector of covariates follows a multivariate standard normal distribution. 
	Suppose $(y, \mathbf{x})$ follows a GLM with parameter $\boldsymbol{\beta}$  with an intercept, denote $\boldsymbol{\beta}=(\beta_0, \boldsymbol{\beta}_*)$ and $\mathbf x =(1,\mathbf{x}_*)$, and suppose  $\mathbf x_* \sim \mathcal N_p(\mathbf 0, I)$. Then $\mathbf x^{\top}\boldsymbol{\beta}  \sim \mathcal N_1(\beta_0, \left\|\boldsymbol\beta_*\right\|^2)$.  For this type of covariates and for Poisson and logistic regression, we generated a grid of values of $\beta_0$ and $\left\|\boldsymbol\beta_*\right\|$ and we computed the ABP for each value. As mentioned in Remark \ref{obs:ABP}, the ABP depends highly on the probability that $y=0$ so, in Figure \ref{fig:bps}, we give a plot of the ABP vs $p_0=\mathbb P(y=0)$. Notice that the breakdown point is near $0.5$ when $p_0 = 0$ and slowly decreases to $0$ as $p_0 \rightarrow 1$. This behavior when $p_0 \rightarrow 1$ is explained in Remark \ref{obs:ABP}.  
	\begin{center}
		\begin{figure}[H]
			\begin{tabular}{cc}
				\includegraphics[width=0.48\textwidth]{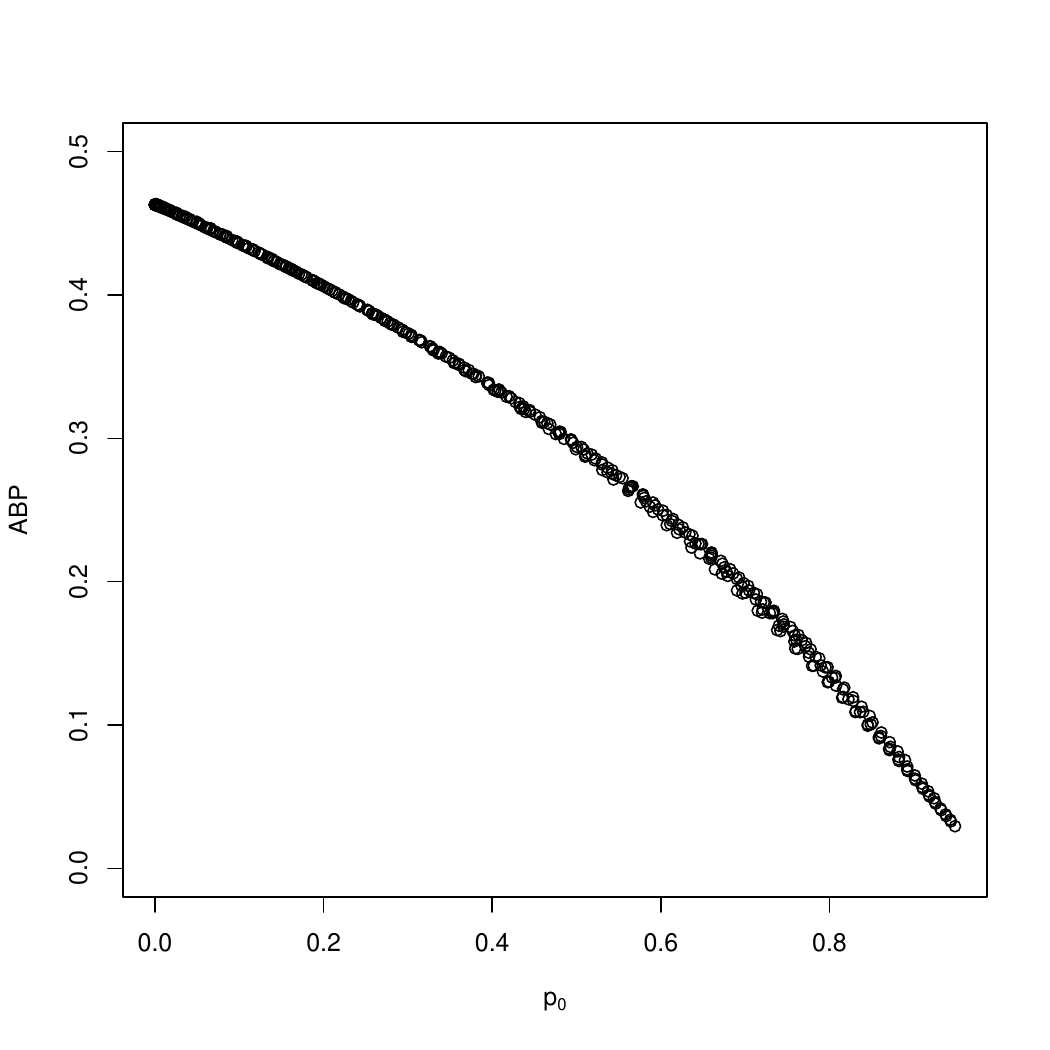} &
				\includegraphics[width=0.48\textwidth]{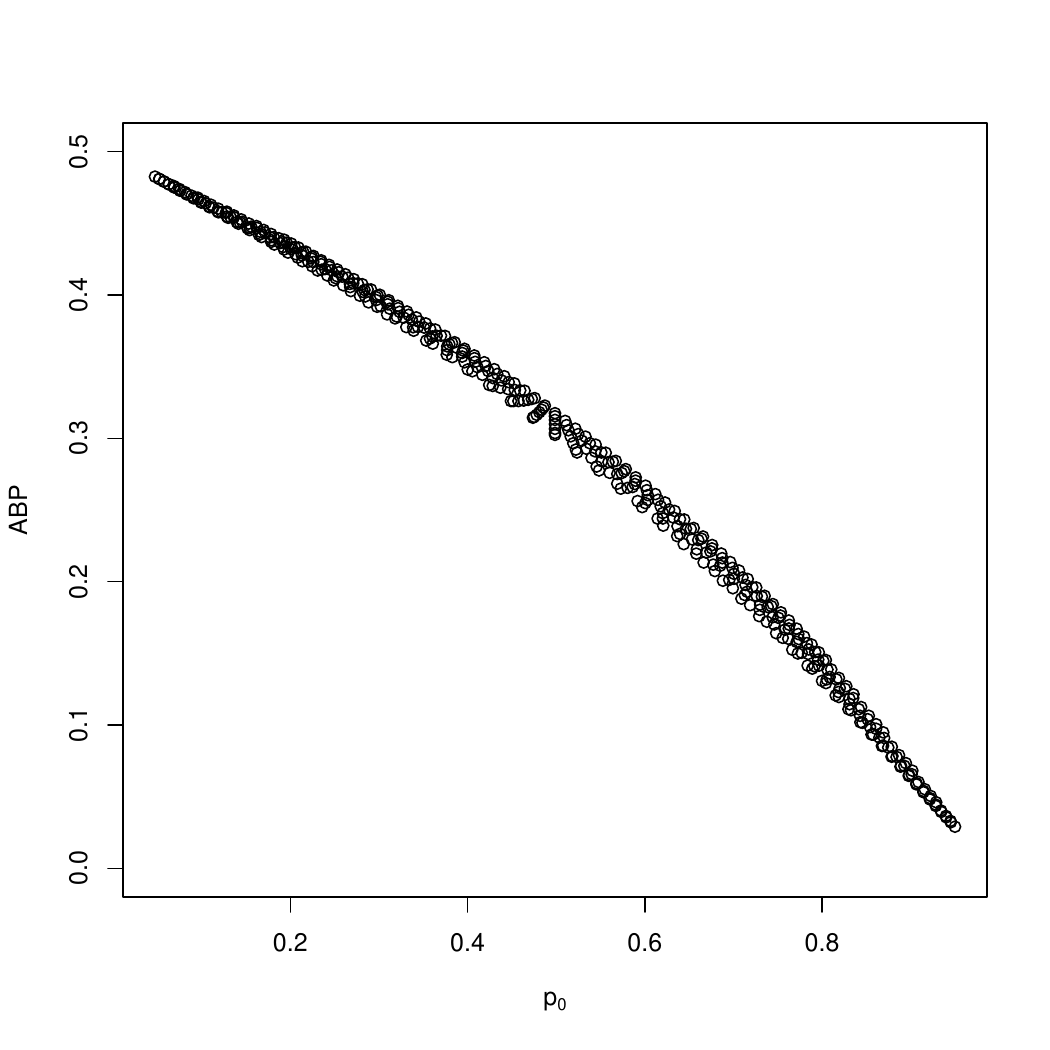}\\
			\end{tabular}
			\caption{Lower bound for ABP of MNQPIT-estimator for Poisson regression (left) and for logistic regression (right) when the vector of covariates have a standard multivariate normal distribution with an intercept, plotted  vs $p_0=\mathbb P(y=0)$.}
			\label{fig:bps}
		\end{figure}
	\end{center}
	\section{Implementation}\label{sec:implementation}
	In our simulations and examples we used an approximation to the optimal $\rho$-function proposed in \cite{yohai1997optimal}, modified as proposed in \cite{konis2021optimal}. See \cite{maronna2019robust} for details.
	In the case of linear models with normal errors, it is known that M-estimators defined using $\rho$-functions in the   \cite{yohai1997optimal} family, have asymptotic minimax bias for different levels of contamination.  We chose  a $\rho$ - function in this family calibrated for 95$\%$ efficiency under the normal model.

	We use a polynomial approximation, similar to the one described in  \cite{martin2023polynomial}. 
	Since our estimator requires a third derivative, we had to ask more conditions of our polynomial and after some trial and error we found that a polynomial of grade 16 gives a good approximation with three derivatives  to the  Konis and Martin $\rho$ - function with $95\%$ of efficiency. The chosen polynomial has the form $q(x)=\sum_{i=1}^8 a_i x^{2i}$. {The values of the coefficients, rounded to 6 decimal digits are $a_1=0.178663$, $a_2=-0.100082$, $a_3=0.096699$, $a_4=-0.043349$, $a_5=0.010057$, $a_6=-0.001250$, $a_7= 0.000079$ and $a_8=-0.000002$
	}. 
	{The $\rho-$function is then defined for each value of $c>0$, as
		\begin{equation*}
			\rho(x) = \left\{ \begin{array}{cc} q(x/c) &\text{ if } |x| \leq c\\
				1 &\text{ if } |x| > c \end{array}\right.
		\end{equation*}
	}
	Polynomials of smaller degree do not give satisfactory approximations to our $\rho$-function, while the chosen polynomial's graph is very similar to the optimal $\rho$'s. See Figure \ref{fig:rhoapprox}. Note that this function satisfies Assumptions \ref{as:rhoeven} to \ref{as:rhocont}.  The calibration constant $c$ was chosen to ensure an efficiency of $0.95$ in the case of the normal linear model. See \cite{maronna2019robust}.
	
	The method was implemented in R (\cite{R2023}).
	The minimum that defines the function $m$ and the minimum that defines MNQPIT and WMNQPIT-estimators were computed using the R function {\tt optim} with method  {\tt "L-BFGS-B"}. The initial estimate for Poisson regression was computed by a subsampling algorithm similar to the one used in \cite{valdora2014robust} and implemented in R package {\tt robustbase} (\cite{maechler2024robustbase}). For logistic regression, the initial estimate was computed by the weighted maximum likelihood method, using the function {\tt logregWML} from package {\tt RobStatTM}; see \cite{yohai2023robstattm}. 
	
	For computing WMNQPIT-estimators, we used the weights computed by the function {\tt logregWML} from package {\tt RobStatTM}. 
	\begin{center}
		\begin{figure}[H]
			\includegraphics[width=\linewidth]{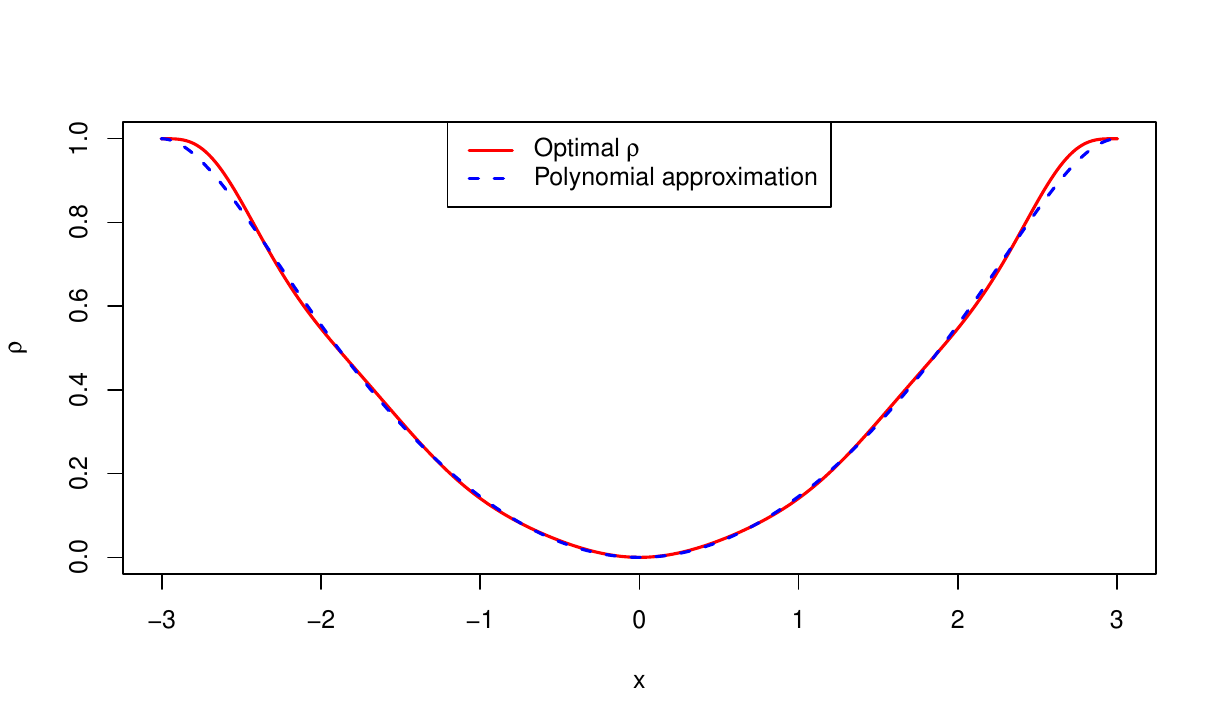} 
			\caption{Optimal $\rho$ and its polynomial approximation.}
			\label{fig:rhoapprox}
		\end{figure}
	\end{center}
		%
	\section{Monte Carlo Study}\label{sec:montacarlo}
	
	We performed a Monte Carlo study for Poisson and logistic regression to compare the performance of the proposed MNQPIT-estimators to classical methods and to some of the existing robust proposals with and without contamination with outliers. 
	
	In each case, we generated $N=1000$ samples of size $n=100$ following the corresponding GLM with parameter $\boldsymbol{\beta}_0$ and we computed several estimators for comparison. 
	To evaluate their performance we computed the empirical mean squared error (MSE) for each estimator $\fhat{\boldsymbol{\beta}}$ as
	\begin{equation}
		\operatorname{MSE}\left(\fhat{\boldsymbol{\beta}}\right)=\frac{1}{N} \sum_{i=1}^N ||\fhat{\boldsymbol{\beta}}_i - \boldsymbol\beta_0||_2^2,
		\label{eq:MSEdef}\end{equation}
	where $\fhat{\boldsymbol{\beta}}_i$ is the result of the estimator $\fhat{\boldsymbol{\beta}}$ in replication $i$ and $||.||_2$ denotes the euclidean norm.
	
	For clean samples, we computed the efficiency of the estimators with respect to ML by dividing their MSE by that of the ML-estimator.
	
	\subsection{Poisson regression}
	The covariate vector was chosen as $(1,\mathbf x)$, where $\mathbf{x}$ is a 5-dimensional standard normal ($ N_5(0,I)$).
	We considered three  GLM models, generated with three different values of $\boldsymbol{\beta}_0$: 
	$\boldsymbol{\beta}_{0,1}=(2, 1, 0, 0, 0, 0)$,
	$\boldsymbol{\beta}_{0,2}=(1.7, 1/3, 0, 0, 0, 0)$ and 
	$\boldsymbol{\beta}_{0,3}=(1.5, 0.1, 0.1, 0.1, 0.1, 0.1)$.
	For each sample we computed the maximum likelihhod (ML) estimator, the MT-estimator proposed in \cite{valdora2014robust}, the HQL-estimator proposed in \cite{cantoni2001robust} and the MNQPIT-estimator.
	All the estimators were computed in R and were calibrated to get approximately $95\%$ efficiency in the simulations.   ML-estimators were computed using the function {\tt glm}, MT and HQL estimators were computed using the function {\tt glmrob}, from package {\tt robustbase} (\cite{maechler2024robustbase}), with the options {\tt method = MT} and {\tt method = Mqle}, respectively. The code used to compute MNQPIT-estimators can be found in the Supplement to this paper.
	
	To evaluate the resistance to outliers of the estimators, we contaminated the samples by replacing $10\%$ of the observations with outliers of the form $(1, \mathbf x_0, y_0)$, with $\mathbf x_{0}=(3,0,0,0,0)$ and several values of $y_0$. The values of $y_0$ were chosen so as to identify the maximum MSE for the MNQPIT-estimator. 
	This simulation setting is similar to the one studied in  \cite{valdora2014robust}. 
	
	In all simulation settings we computed the MSE of all the estimators, as defined in \eqref{eq:MSEdef}.
	For clean samples we also computed the efficiency with respect to ML, as the quotient of the MSE of each estimator and the MSE of ML; see
	Table \ref{table:effpoisson}. Figures \ref{fig:simusconconta1} to \ref{fig:simusconconta3} give the MSE of the estimators under contamination, as a function of $y_0$.
	\begin{table}[H]
		\centering
		\begin{tabular}{rrrr}
			\hline
			& MT & HQL&MNQPIT \\ 
			\hline
			$\boldsymbol{\beta}_{0,1}$ & 0.96 & 0.94 & 0.93 \\
			$\boldsymbol{\beta}_{0,2}$ & 0.97 & 0.94 & 0.94\\
			$\boldsymbol{\beta}_{0,3}$ & 0.96 & 0.95 & 0.95 \\ 
			\hline
		\end{tabular}\caption{Efficiencies with respect to ML in uncontaminated samples from the Poisson regression model.}\label{table:effpoisson}	
	\end{table}
	\begin{center}
		\begin{figure}[H]\begin{center}
				\includegraphics[width=0.70\linewidth]{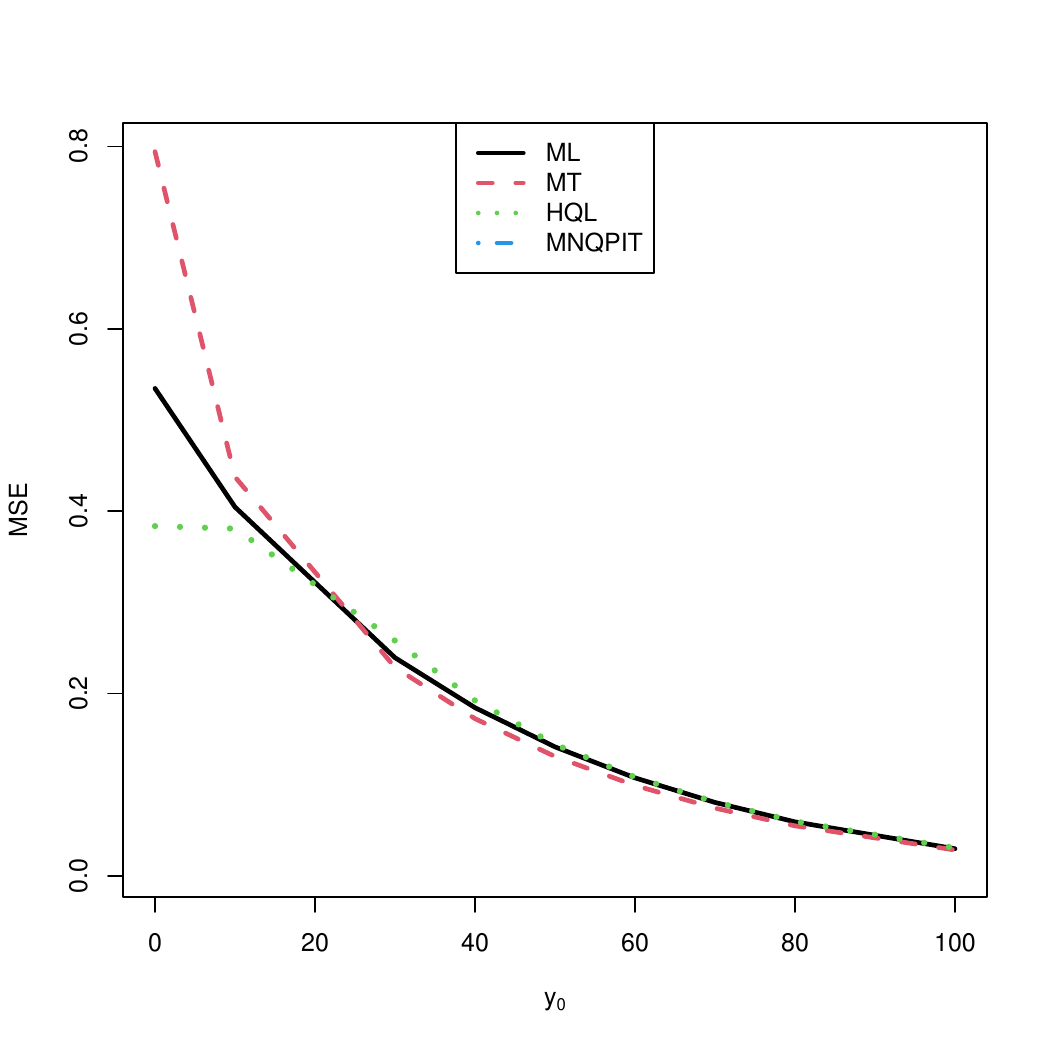}\\
				\caption{MSE for contaminated samples for $\boldsymbol{\beta}_{0}=\boldsymbol{\beta}_{0,1}$.}
				\label{fig:simusconconta1}
			\end{center}
		\end{figure}
	\end{center}				
	
	\begin{center}
		\begin{figure}
			\begin{center}
				\includegraphics[width=0.70\linewidth]{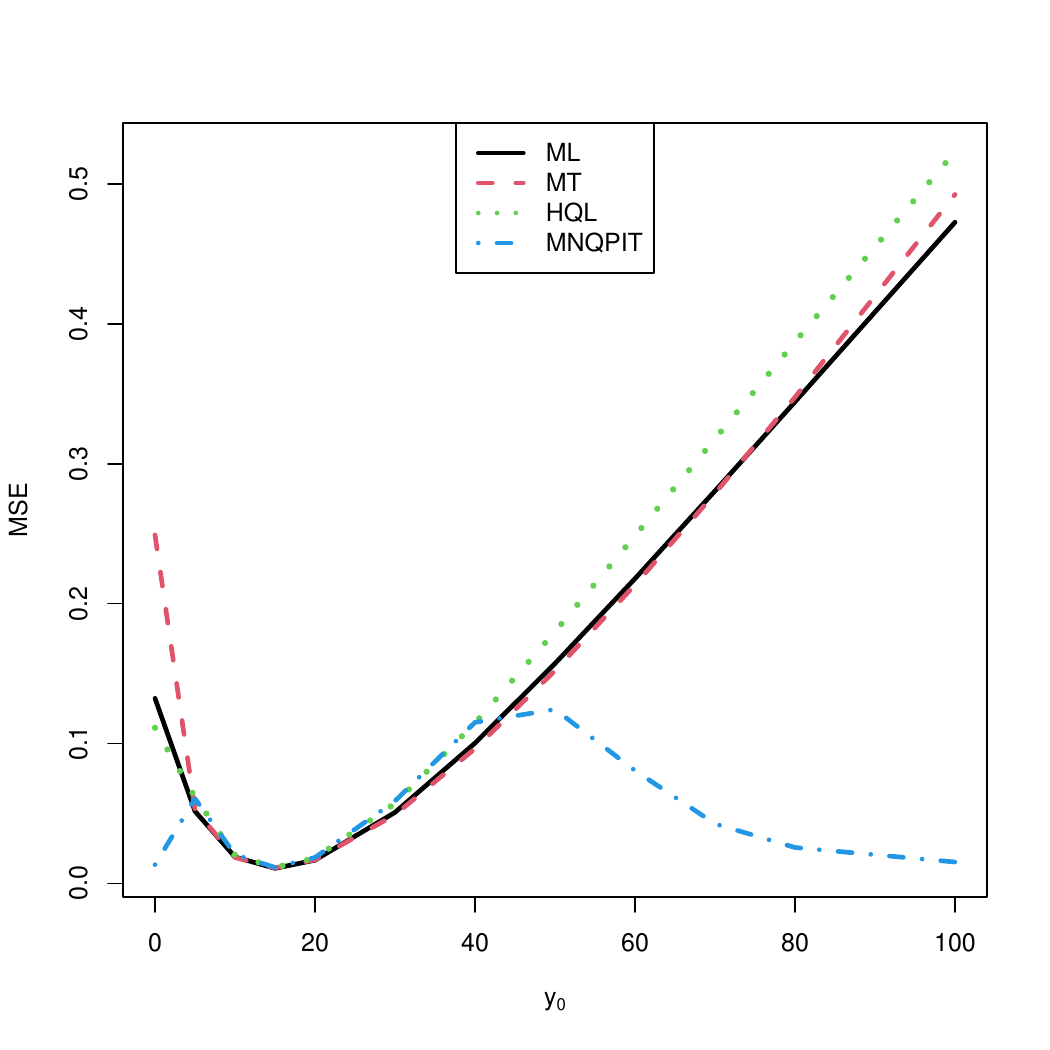}
				\caption{MSE for contaminated samples for 		$\boldsymbol{\beta}_{0}=\boldsymbol{\beta}_{0,2}$ }
			\end{center}
			\label{fig:simusconconta2} 
		\end{figure}		
	\end{center}
	
	\begin{center}		
		\begin{figure}
			\begin{center}
				\includegraphics[width=0.70\linewidth]{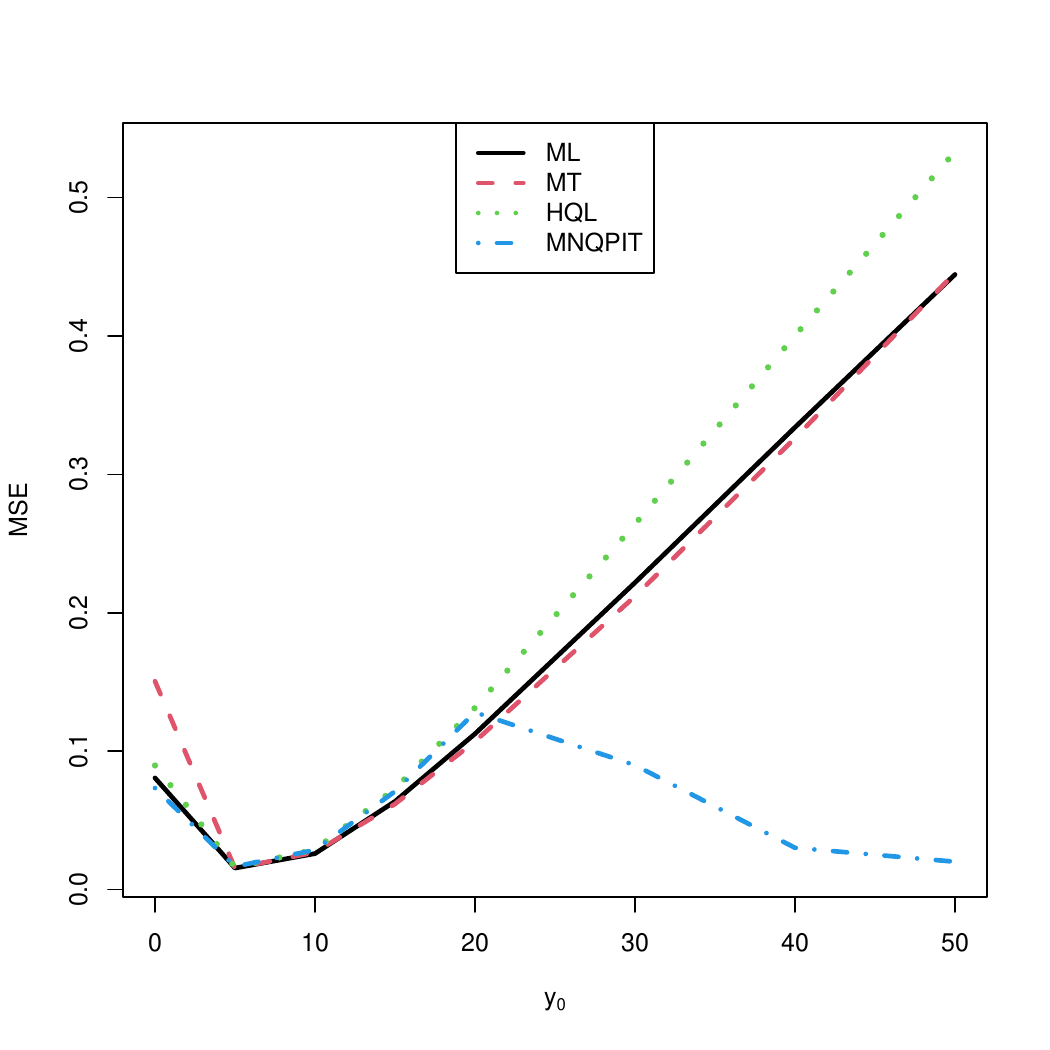}\\
			\caption{MSE for contaminated samples for 		$\boldsymbol{\beta}_{0}=\boldsymbol{\beta}_{0,3}$ } 	\label{fig:simusconconta3}
		\end{center}
	\end{figure}
\end{center}

Note that the minimum MSE for all estimators is obtained when $y_0$ is nearest to $\mu_0=\mathbb{E}\left(y|\mathbf{x}_0\right)$. This expected value is approximately $148.41$, when $\boldsymbol{\beta}_0=\boldsymbol{\beta}_{0,1}$, $14.88$, when $\boldsymbol{\beta}_0=\boldsymbol{\beta}_{0,2}$ and $6.04$, when $\boldsymbol{\beta}_0=\boldsymbol{\beta}_{0,3}$. As $y_0$ moves away from $\mu_0$, the MSE of the estimators increases. MNQPIT's MSE increases until it reaches a maximum of $0.11$ for $\boldsymbol{\beta}_{0,1}$, $0.12$ for $\boldsymbol{\beta}_{0,2}$ and $0.13$ for $\boldsymbol{\beta}_{0,3}$. On the other hand, the maximum MSE found for the other estimators 
is much higher.

The difference between the results obtained in this simulation study and that presented in \cite{valdora2014robust} is that MT has been calibrated for a higher efficiency to make it comparable with MNQPIT.

\subsection{Logistic regression}

We generated samples of size $n=100$ following a logistic GLM with 5-dimensional covariates $\mathbf{x}\sim N(\mathbf 0, \mathbf I)$, with an intercept and regression parameter is
$\boldsymbol{\beta}_{0}=(0, 2, 2, 0, 0, 0)$.  
For each sample we computed the maximum likelihhod (ML) estimator, the BY-estimator proposed in \cite{bianco1996robust} with the implementation proposed in \cite{croux2003implementing}, the WBY-estimator proposed in \cite{croux2003implementing}, the WML-estimator, which is simply an ML-estimator computed after penalizing high leverage observations with weights as defined in \eqref{eq:weight}, the MNQPIT-estimator and the WMNQPIT-estimator proposed in this paper.
All the estimators were computed in R. ML-estimators were computed using the function {\tt glm}, WML, BY and WBY estimators were computed using the functions {\tt logregWML}, {\tt BYlogreg} and {\tt WBYlogreg}, respectively, from package {\tt RobStatTM}. Finally, the R code to compute MNQPIT and WMNQPIT-estimators can be found in \texttt{https://github.com/mvaldora/mnqpit-glm-poisson-logistic}.

We first performed a simulation study without outliers. 
The efficiencies of the estimators  with respect to the maximum likelihood estimator (ML) are given in Table \ref{table:efflogMNQPIT}. 
\begin{table}[ht]
	\centering
	\begin{tabular}{rrrrrrrrr}
		\hline
		WML & BY & WBY 
		&MNQPIT &WMNQPIT \\ 
		\hline
		0.93 & 0.93 & 0.87 & 1.11 & 1.03 \\ 
		\hline
	\end{tabular}\caption{Efficiencies with respect to ML for uncontaminated samples from the logistic regression model.
	}
	\label{table:efflogMNQPIT}
\end{table}
MNQPIT and WMNQPIT-estimators were calibrated using the constant necessary to obtain $95\%$ efficiency in the normal GLM. However, the efficiency we obtained is much higher, namely 1.11 for MNQPIT and 1.03 for WMNQPIT.
This means that for these samples, MNQPIT and WMNQPIT perform better than ML when there are no outliers.
We then tried to calibrate BY and WBY estimators to get the same efficiency in the  simulations. However, we could not obtain such a high efficiency for BY or WBY, so we chose calibrating constants for 0.93 and 0.87 efficiencies respectively.

{
	We then contaminated the samples  by adding 10 outliers of the form $(\mathbf{x}_0, y_0)$, where $\mathbf x_0=( z_0,z_0,z_0,z_0,z_0)$ for $z_0 \in \{0.5, 1.5, 2.5, 3.5, 4.5\}$ and $y_0=0$. 
	Note that,  if $\mathbf x_0=(z_0, z_0,z_0,z_0,z_0)$, $\mathbb P(y=1|\mathbf x=\mathbf x_0) = \exp(2z_0 + 2z_0)/(1+\exp(2z_0 + 2z_0))$ takes the values 
	$0.8807971, 0.9975274, 0.9999546, 0.9999992, 1.0000000$. Putting $y_0=0$, makes these added values atypical.
	The larger the value of $z_0$, the higher the leverage of the outlier and the more atypical the value of $y_0$.} 
Figure  \ref{fig:mses_cont_binomial} gives the MSE of all estimates for logistic regression as a function of $z_0$. 
\begin{figure}[H]
	\begin{center}
		\includegraphics[width=0.7\linewidth]{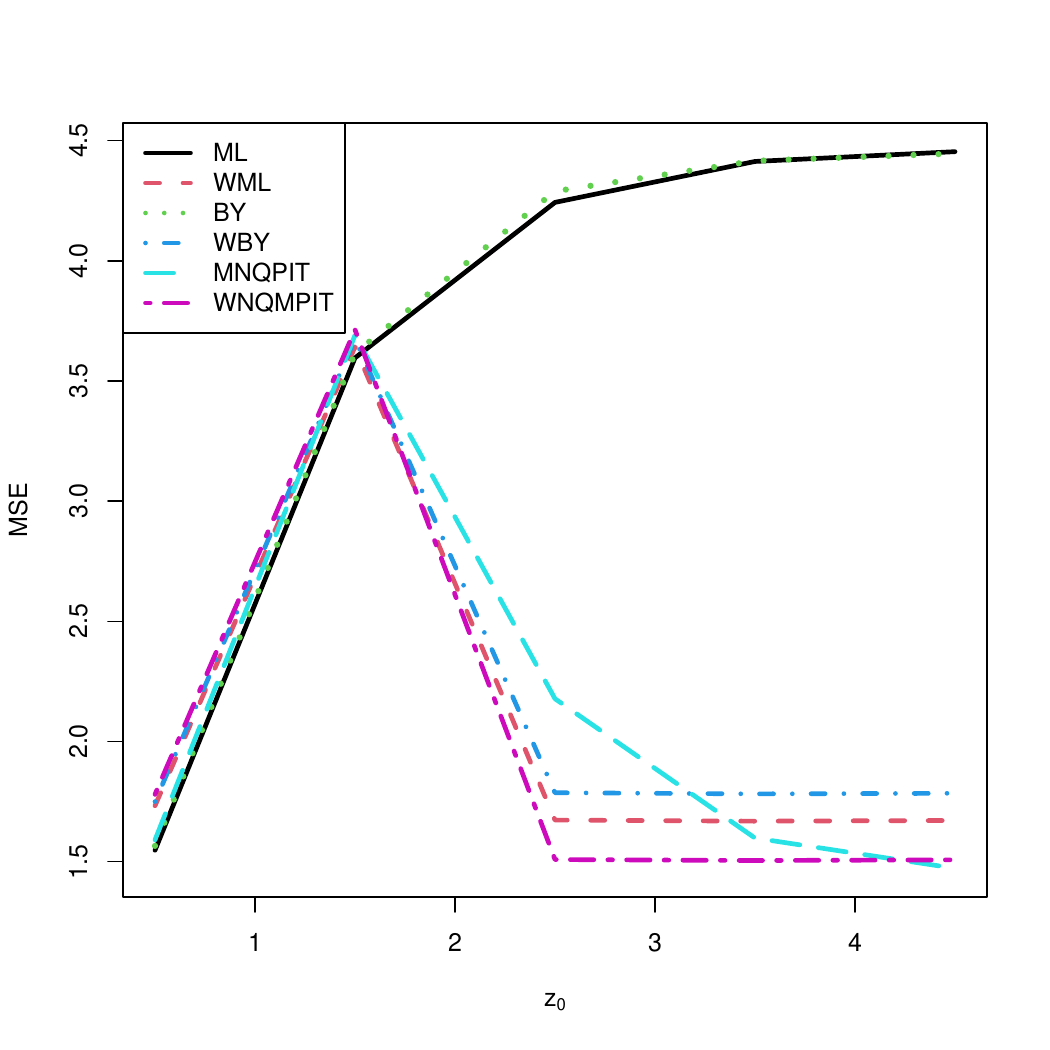}\\
	\end{center}
	\caption{Mean squared errors for contaminated samples in logistic regression} 
\label{fig:mses_cont_binomial}
\end{figure}

Note that MNQPIT and WMNQPIT-estimators have a very high efficiency for these models and at the same time a very high robustness. WBY, BY and WML give a much lower efficiency in clean samples and a slightly lower maximum MSE for contaminated samples. However, for  extreme outliers, the MSE of MNQPIT and WMNQPIT is smaller.

\section{Real data examples}\label{sec:realdata}
\subsection{Example: Crohn disease data}\label{subsec:crohn}
In this section we analyze the Crohn disease data set ({\tt CrohnD}) from R package {\tt{robustbase}}. The response variable is the number of adverse events and there are seven covariates related to height, weight, country and treatment received. The number of observations is $n=117$.

We fit a Poisson GLM by four methods: ML, HQL, MT and MNQPIT and, for each method, we computed the RQ residuals. These residuals are similar to those introduced in \cite{dunn1996randomized} and  implemented in the Dharma R package (\cite{hartig2024dharma}).  For each observation $(\mathbf x_i, y_i)$ and each estimator $\fhat{\boldsymbol\beta}$, the RQ residual is defined as
$r_i(\fhat{\boldsymbol\beta}) =\Phi^{-1} \left(F(y_i, \fhat \theta_i) - u f(y_i, \fhat \theta_i)\right)$, where $\fhat \theta_i=\exp(\mathbf x_i^\top \fhat{\boldsymbol\beta})$ and $u\sim\mathcal U(0,1)$, independent of $(y_i, \mathbf x_i)$.
The computation of the RQ residuals was repeated using 100 diferent seeds. In each repetition, the same seed was used for all the methods.
Let $RQ_{ik}$ ,  $1\leq i \leq 117$,
$1\leq k\leq 100$, be the 
$i$-th  RQ   residual for the $k$-th replication and define ${ARQ}_i$ as  
the mean $RQ_{ik}$ for $k = 1, \dots, 100$. These averaged randomized quantile (ARQ) residuals were used to measure the goodness of the fits. 
For each estimator we computed the quantiles of the absolute values of the ARQ residuals. Figure \ref{fig:quantilesdunnres} gives a comparative plot of the $0.1$ to $0.9$ quantiles, while Table \ref{table:mediansdunnrescrohn} gives their median. 
\begin{figure}
\centering
\includegraphics[width=0.7\linewidth]{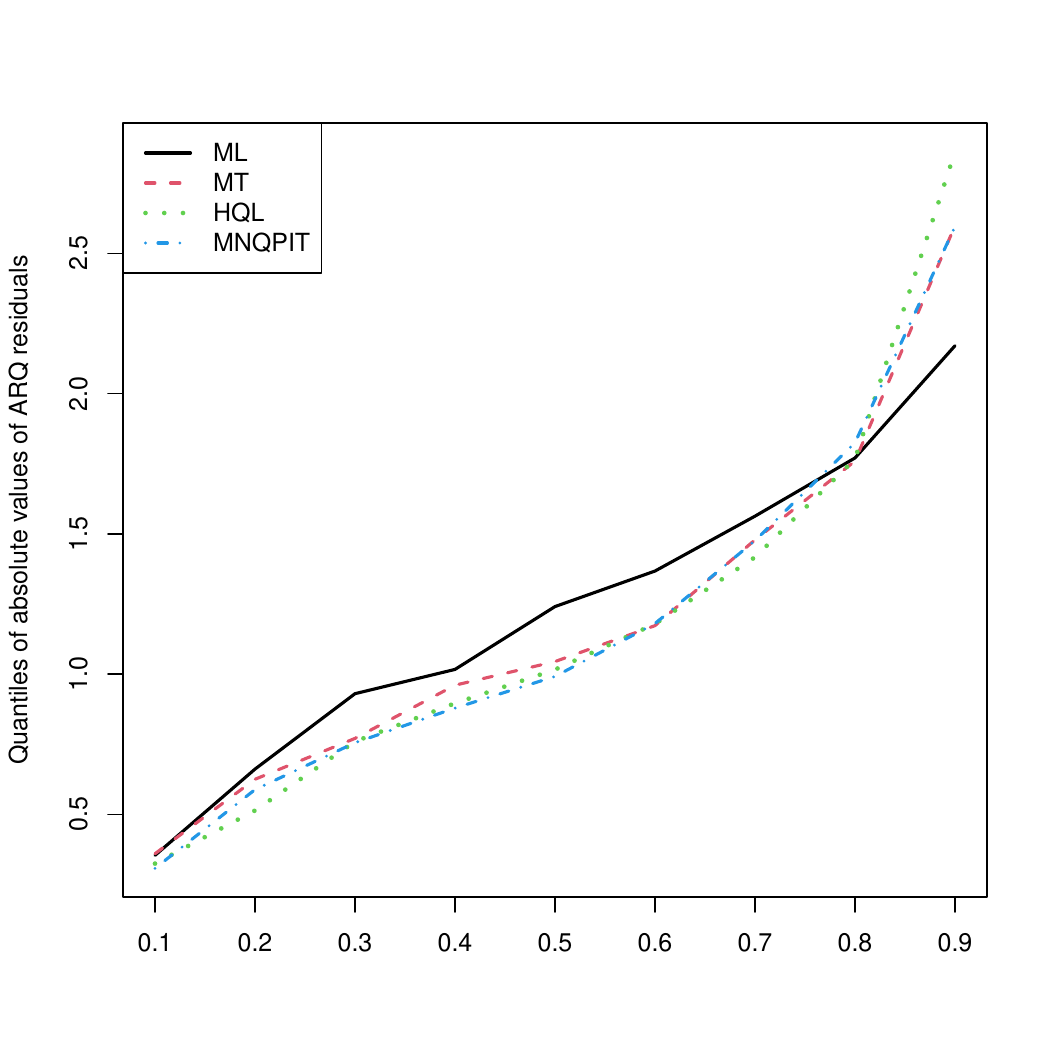}
\caption{Quantiles for the ARQ residuals for Crohn disease data}
\label{fig:quantilesdunnres}
\end{figure}

Figure \ref{fig:quantilesdunnres} does not show a big difference among the absolute values of the ARQ residuals of HQL, MT and MNQPIT methods while Table \ref{table:mediansdunnrescrohn} shows that the median of the absolute values of the ARQ residuals is a little smaller for MNQPIT.

\begin{table}[ht]
\centering
\begin{tabular}{rrrr}
\hline
ML & HQL & MT & MNQPIT \\ 
\hline
1.24 & 1.04 & 1.02 & 0.99 \\ 
\hline
\end{tabular}\caption{Median of the absolute values of the ARQ residuals for Crohn disease data.}\label{table:mediansdunnrescrohn}
\end{table}



\subsection{Leukemia data}
In this section we analyze the Leukemia data set ({\tt{leuk.dat}}) from R package {\tt{RobStatTM}}. See also Example 7.1 in \cite{maronna2019robust}.
The dataset contains information about 33 leukemia patients.
The response y is a binary variable that equals 1 if the patient survives more than 52 weeks and 0 otherwise. The covariates are white blood cell count (wbc) and the presence or absence of a certain morphological characteristic in the white cells (ag).

We fitted a logistic regression model by the six methods considered in the simulation study and, for each estimator, we computed  the ARQ residuals  averaged over 100 diferent seeds, as in Section \ref{subsec:crohn}.

Figure \ref{fig:boxplotresdunnleukemia} gives boxplots of the absolute values of the ARQ residuals. Note that all methods but ML detect a clear outlier, which is observation 15. Table \ref{table:coef_leukemia} gives the estimated coefficients by the six methods considered. The coefficients corresponding to covariate wbc are multiplied by $10^4$ for better visualization. In the last column we added the maximum likelihood estimator computed without observation 15. Note the similarity of WML, MNQPIT and WMNQPIT to ML$^*$. 

\begin{figure}
\centering
\includegraphics[width=0.7\linewidth]{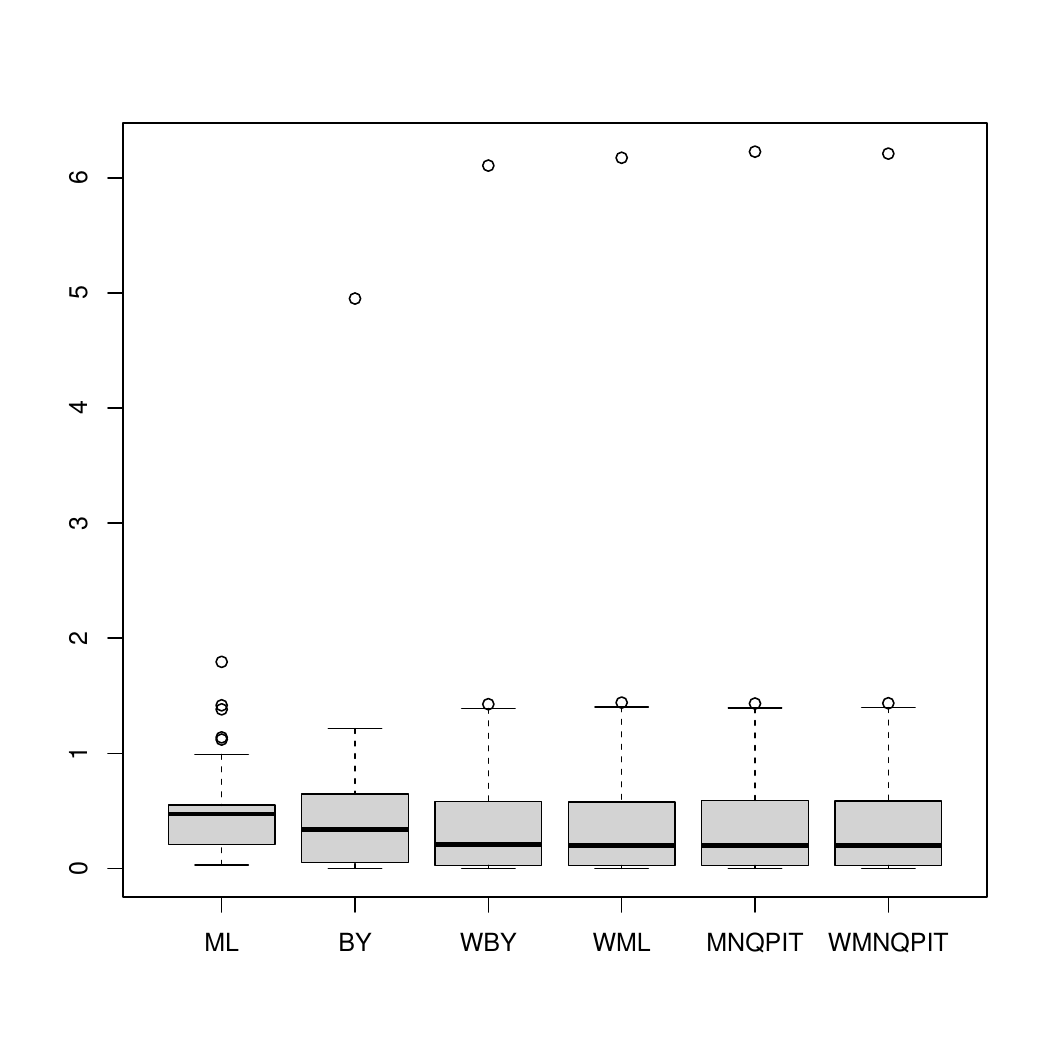}
\caption{Boxplots of absolute values of the ARQ residuals for leukemia data}
\label{fig:boxplotresdunnleukemia}
\end{figure}

\begin{table}[H]
\centering
\begin{tabular}{rrrrrrrr}
\hline
& ML & BY & WBY & WML & MNQPIT & WMNQPIT & ML$^*$ \\ 
\hline
(Intercept) & -1.3073 & 0.1402 & 0.2073 & 0.2116 & 0.2116 & 0.2116 & 0.2119 \\ 
wbc & -0.0318 & -0.1560 & -0.2305 & -0.2354 & -0.2354 & -0.2354 & -0.2354 \\ 
ag & 2.2611 & 1.6951 & 2.5053 & 2.5579 & 2.5579 & 2.5579 & 2.5581 \\ 
\hline
\end{tabular} 
\caption{Estimated coefficents for leukemia data}\label{table:coef_leukemia}
\end{table}
\begin{table}[ht]
\centering
\begin{tabular}{rrrrrrr}
\hline
& ML & BY & WBY & WML & MNQPIT & WMNQPIT \\ 
\hline
&0.4716 & 0.3402 & 0.2064 & 0.1995 & 0.2009 & 0.2004\\
\hline
\end{tabular}
\caption{Median absolute values of the ARQ residuals for leukemia data}\label{table:medianas_leukemia}
\end{table}

Figure  \ref{fig:quantilesdunnresleuk} gives the $0.1$ to $0.9$ quantiles of the absolute values of the ARQ residuals.
This figure shows that the quantiles of all the methods but ML and BY are very similar.

Table \ref{table:medianas_leukemia} gives the median absolute values of the ARQ residuals for all the estimators. In this table we can note a small difference: the smallest median is that of  WML method, followed by WMNQPIT, MNQPIT and WBY methods. 

The residual analysis suggests that WML, MNQPIT and WMNQPIT are less affected by the outlier than BY and ML, while the comparison of the coefficients suggests that WML, MNQPIT and WMNQPIT use the information contained in the non-outlying observations more efficiently than WBY, maintaining approximately the same level of robustness.

\begin{figure}
\centering
\includegraphics[width=0.7\linewidth]{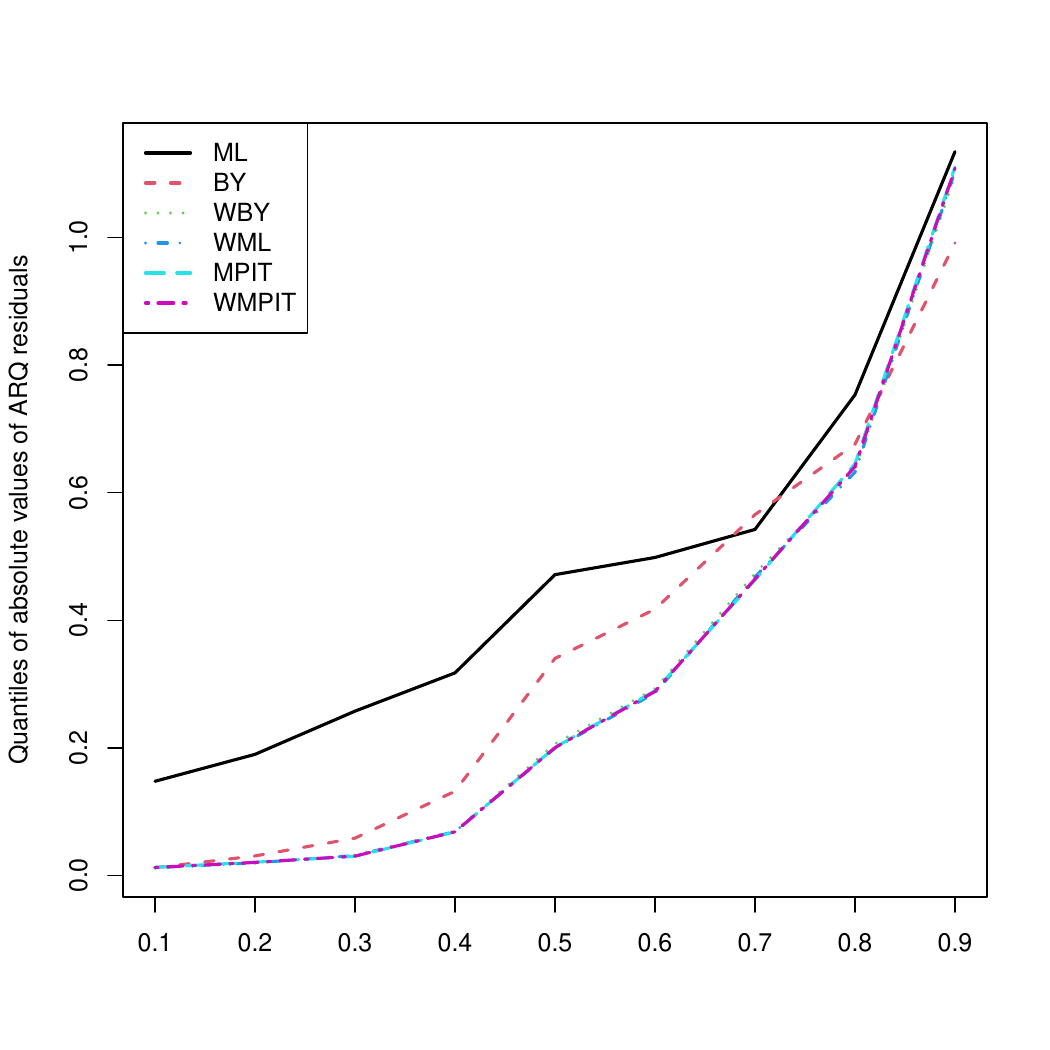}
\caption{Quantiles of the absolute values of the ARQ residuals for leukemia data}
\label{fig:quantilesdunnresleuk}
\end{figure}



\section{Final remarks}

Robust estimators for linear regression models are widely known and used by practitioners. Extending these techniques to GLMs is a natural and necessary step. However, while there are many proposed robust methods for GLMs, there is still no method that is  recognized as the best. 
This slows down their use in real applications.

In this paper we have presented a new method for robust estimation in GLMs based on the NQPIT. By using the NQPIT to "normalize" the residuals, this method builds a bridge between the well-developed theory of robust regression for normal errors and the wider world of GLMs, offering a unified framework for robust estimation across different GLM families.
We have shown the very good performance of this method in the case of binomial and Poisson models. 
Other frequently used GLMs 
are negative binomial and beta GLMs. 
How to extend this method to these models and others that include a dispersion parameter, is not obvious and will be the subject of further work.

\section*{Funding}
This research was partially supported by Grants 20020170100330BA and 20020220200037BA
from the University of Buenos Aires.
\appendix
\section{Appendix}
 This appendix contains the proofs of the  lemmas and theorems stated in Section 4. Section \ref{sec:Scons} contains all the proof related to  the consistency of  WMNQPIT estimators, Section \ref{sec:Sasnorm} the proofs related to its asymptotic normality and Section \ref{sec:SBP} the proofs of two theorems related to its breakdown point.
\subsection{Consistency}\label{sec:Scons}
We start by proving seven lemmas required for the proof of Theorem 1. \begin{lemma}\label{lem:limt} If Assumptions {\upshape{A2}} and {\upshape{A3}} 
	hold, then
	$\lim_{\theta\rightarrow\theta_1} t(s,\theta)=+\infty$ and $\lim_{\theta\rightarrow\theta_2} t(s,\theta)=-\infty$ for all $s$ in the support of $y$ with $0<  F(s, \theta)<1$.
\end{lemma}
\begin{proof}
	First note that the set $C=\{s\in\mathbb R \text{ such that } 0<  F(s, \theta)<1 \}$ is independent of $\theta$ by Assumption {\upshape{A3}}. 
	If $F( . , \theta)$ is continuous, then the result follows from {\upshape{A3}}. 
	Let $F( . , \theta)$ be a discrete distribution  function. Let $s$ be an element of the support of $y$
	and suppose that there
	exists $s_{1}<s$ also belonging to the support of $y,$ that is, such
	that $p(s_{1},\theta)>0$. Then
	$$F(s_1, \theta)\leq F(s,\theta) - p(s, \theta) \leq F(s,\theta) - 1/2 p(s, \theta) < F(s, \theta).$$
	Assumptions {\upshape{A2}} and {\upshape{A3}} 
	imply that both ends of the above inequality converge to 1 when $\theta\rightarrow\theta_1$ and to $0$ when $\theta\rightarrow\theta_2$. This means that so does the middle term. This implies that  $\lim_{\theta\rightarrow\theta_1} t(s,\theta)=+\infty$ and $\lim_{\theta\rightarrow\theta_2} t(s,\theta)=-\infty$ 
	for all $s$ in the
	support of$_{\ \ }y$ \ such that there exists $s_{1}$ satisfying
	$0<F(s_{1},\theta)<F(s,\theta)<1$.
	
	Suppose now that there is no $s_{1}<s$ such that $F(s_{1},\theta)>0.$ Then,
	in that case, $F(s,\theta)=p(s,\theta)$ and then $t(s,\theta)=\Phi
	^{-1}(F(s,\theta)/2)$ and the lemma follows from
	Assumption {\upshape{A3}}.
\end{proof}
\begin{lemma}\label{teo:wmpitareFC}
	Let $(y, \mathbf{x})$ follow a GLM with parameter $\boldsymbol{\beta}_0$ and link function $g$. Under Assumptions {\upshape{A4}} and {\upshape{A5}}, 
	the WMNQPIT-estimator of $\boldsymbol{\beta}_0$ is Fisher consistent.
\end{lemma}
\begin{proof}\begin{align*}
		&\mathbb E_{\boldsymbol{\beta}_0}\left(\rho\left(t\left(y,m\left(g^{-1}\left( \mathbf{x}^\top \boldsymbol{\beta}\right)\right)\right)\right) w\left(\mathbf{x}, \boldsymbol{\mu}_0, \boldsymbol{\Sigma}_0\right)\right)\\
		&\quad  = \mathbb E\left[ \mathbb E_{\beta_0}\left(\rho\left(t\left(y, m \left(g^{-1}\left( \mathbf{x}^\top \boldsymbol{\beta}\right)\right)\right)\right) \mid \mathbf{x}\right) w\left(\mathbf{x}, \boldsymbol{\mu}_0, \boldsymbol{\Sigma}_0\right)\right].
	\end{align*}
	Since $\mathbb E_{\boldsymbol{\beta}_0}\left(\rho\left(t\left(y,m\left(g^{-1}\left( \mathbf{x}^\top \boldsymbol{\beta}\right)\right)\right)\right) \mid \mathbf{x}\right)$ is minimized when $\boldsymbol{\beta}=\boldsymbol{\beta}_0$ for all $\mathbf{x}$, then $\mathbb E_{\boldsymbol{\beta}_0}\left(\rho\left(t\left(y,m\left(g^{-1}\left( \mathbf{x}^\top \boldsymbol{\beta}\right)\right)\right)\right) w\left(\mathbf{x}, \boldsymbol{\mu}_0, \boldsymbol{\Sigma}_0\right)\right)$ is also minimized when $\boldsymbol{\beta}=\boldsymbol{\beta}_0$. Therefore WMNQPIT-estimators are Fisher consistent.
\end{proof}
\begin{lemma}\label{lemma:increasing} {Let $F(y,\theta)$ be a c.d.f corresponding to a discrete or continuous random variable} and let $t$ be defined by equation {\upshape(4)}. 
	Then, under Assumption {\upshape{A2}}, 
	$t(y, \theta)$ is increasing in $y$ and descreasing in $\theta$.
\end{lemma}
\begin{proof} {The continuous case is trivial. Let us consider the discrete case}. 
	Let $y_1 < y_2$, then 
	$F(y_1, \theta) + p(y_2,\theta) \leq F(y_2,\theta)$. Therefore,
	
	$t(y_1, \theta) = F(y_1, \theta) - 1/2 p(y_1, \theta) \leq F(y_1, \theta) + 1/2 p(y_2, \theta) \leq F(y_1, \theta) + p(y_2, \theta) - 1/2 p(y_2, \theta) \leq
	F(y_2, \theta) - 1/2 p(y_2, \theta) = t(y_2, \theta)$.
	
	Now let $\theta_1 < \theta_2$. To simplify notation, we assume  $F(y,\theta)$ is the distribution of a random variable with rank $\mathbb N_0 = \mathbb N \cup \{0\}$. If its rank is any other countable or finite set, the proof is analogous. By Assumption {\upshape{A2}}, 
	we know that $F(y, \theta_2) \leq F(y, \theta_1)$ and also that $F(y-1, \theta_2) \leq F(y-1, \theta_1)$. We consider two cases:  $p(y, \theta_2) \leq p(y, \theta_1)$ (case 1) and $p(y, \theta_2) > p(y, \theta_1)$ (case 2). In case 1, we get that 
	$$F(y-1, \theta_2)  + \frac{1}{2} p(y, \theta_2) \leq F(y-1, \theta_1)+ \frac{1}{2} p(y, \theta_1),$$ that implies 
	\setcounter{equation}{13}
	\begin{equation}\label{eq:tpmonotona}
		F(y, \theta_2)  - \frac{1}{2} p(y, \theta_2) \leq F(y, \theta_1) - \frac{1}{2} p(y, \theta_1).	
	\end{equation}	
	In case 2 \eqref{eq:tpmonotona}, follows directly from  $F(y, \theta_2) \leq F(y, \theta_1)$.
	The results then follows from the monotonicity of $\phi^{-1}$.
\end{proof}
\begin{lemma}\label{lema:espvart0} 
	Let $y$ be a discrete {or continuous random variable}, let $F(., \theta)$ be its cumulative distribution function and $p(k,\theta)=\mathbb P_\theta(y=k)$.  Let $t_0(y,\theta)=F(y,\theta) - 0.5\, p(y,\theta)$.
	Then  $\mathbb E_\theta (t_0(y,\theta)) =0.5$ and $\mathbb V_\theta (t_0(y, \theta)) =  \frac{1}{12}\left(1 - \mathbb  E_\theta\left(  p(y, \theta) ^2  \right) \right)$.
\end{lemma}
\begin{proof} {If $y$ is a continuous random variable, then $p(s, \theta)=0$ for all $s$, then $t_0(y, \theta)$ is uniformly distributed on $(0,1)$  and the result follows.}
	
	{
		If $y$ is a discrete random variable,} let $t_P$ be as in equation (1), then
	$t_0(y,\theta)= \mathbb E_\theta ( t_P(y, u, \theta)|y)$ and therefore $\mathbb E_\theta (t_0(y,\theta)) = \mathbb E_\theta \left(\mathbb E ( t_P(y, u, \theta)|y) \right)= \mathbb E_\theta( t_P(y, u, \theta))=0.5$.
	\begin{align*}
		\mathbb E_\theta(t_0(y, \theta)^2)& = \mathbb E_\theta\left( F(y, \theta)^2  -  F(y, \theta) p(y, \theta) +
		0.25 \,p(y, \theta) ^2 \right)\\ \nonumber
		& = \mathbb E_\theta \left( \mathbb E\left( F(y, \theta)^2  -  2 u F(y, \theta) p(y, \theta) + u^2 \,p(y, \theta) ^2  +
		0.25 \,p(y, \theta) ^2 - u^2\,p(y, \theta) ^2 | y \right)\right)  \\ \nonumber
		& =  \mathbb E_\theta(t_P(y, u, \theta)^2) + \mathbb E_\theta \left( \mathbb  E\left( 0.25 \,p(y, \theta) ^2 - u^2\,p(y, \theta) ^2 | y \right)\right)  \\ \nonumber
		& = \frac{1}{3} - \frac{1}{12} \mathbb  E_\theta \left(  p(y, \theta) ^2  \right).
	\end{align*} Then 
	\begin{equation*}
		\mathbb V_\theta(t_0(y, \theta)) =  \frac{1}{12}\left(1 - \mathbb  E_\theta\left(  p(y, \theta) ^2  \right) \right).
	\end{equation*}
\end{proof}
\begin{lemma}\label{lemma:epsilon0} Under Assumptions {\upshape A6}, 
	{\upshape A7} 
	and {\upshape A8} 
	and $\epsilon_0$ as in {\upshape A6}, 
	$\mathbb E_\theta\left(\rho\left( t(y,\theta)\right)\right)<1-\epsilon_0$.
\end{lemma}
\begin{proof} Let $c$ be as in Assumption A6. Then, by Assumptions A7  and A8,
	\begin{align}\label{eq:indicadoras} \nonumber
		\mathbb E_\theta\left(\rho \left(t(y,\theta)\right)\right)& = \mathbb E_\theta\left(\rho\left( t(y,\theta)\right) I(\left| t(y,\theta)\right|\leq c)\right) +  \mathbb E_\theta\left(\rho\left( t(y,\theta)\right) I(\left| t(y,\theta)\right|>c)\right)\\ 
		&\leq \rho(c)  + \mathbb P\left( t_0(y,\theta) >\Phi(c)  \right) + \mathbb P\left( t_0(y,\theta) <\Phi(-c)  \right).  
	\end{align}
	By Chebyshev's inequality,  we know that
	\begin{eqnarray}\label{eq:cheby}
		\mathbb P_\theta\left( \left| t_0(y,\theta) - 0.5\right| >\Phi(c) -0.5 \right) <  \frac{\mathbb V(t_0(y, \theta))}{(\Phi(c) -0.5)^2}.
	\end{eqnarray}
	Since
	\begin{eqnarray*}
		\mathbb P_\theta\left( \left| t_0(y,\theta) - 0.5\right| >\Phi(c) - 0.5\right) & = \mathbb P_\theta\left( t_0(y,\theta)  >\Phi(c)  \right)  + \mathbb P_\theta\left( t_0(y,\theta)  < 1- \Phi(c)   \right)   \\
		&  = \mathbb P_\theta\left( t_0(y,\theta)  >\Phi(c)  \right)  + \mathbb P_\theta\left( t_0(y,\theta)  <  \Phi(-c)   \right),
	\end{eqnarray*}
	combining  \eqref{eq:indicadoras},  \eqref{eq:cheby},  Lemma \ref{lema:espvart0} and Assumption A6,  we get that
	\begin{eqnarray*}
		\mathbb E_\theta\left(\rho\left( t(y,\theta)\right)\right) \leq \rho(c) + \frac{\mathbb V(t_0(y, \theta))}{(\Phi(c) -0.5)^2}  \leq  \rho(c)  +  \frac{1}{12 (\Phi(c) -0.5)^2}< 1 -\epsilon_0.
	\end{eqnarray*}
\end{proof}
\begin{cor}\label{cor:epsilon0} Under Assumptions A5 to A8,
	there exists $\epsilon_0$ such that 	$$\mathbb E_\theta\left(\rho \left(t(y,m(\theta))\right)\right)<1-\epsilon_0.$$
\end{cor}
\begin{proof}
	It follows from the fact that	$\mathbb E_\theta\left(\rho \left(t(y,m(\theta))\right)\right)<\mathbb E_\theta\left(\rho \left(t(y,\theta)\right)\right)$ and Lemma \ref{lemma:epsilon0}.
\end{proof}
We will need the following lemma to prove the continuity of the function m.
\begin{lemma}\label{lema:victor}
	Let $F_n$ be a sequence of distribution functions that converges weakly to $F$. Let $a_n$ be a deterministic sequence with $a_n \rightarrow 
	a$. Suppose that $h: \mathbb{R}^2 \rightarrow \mathbb{R}$ is a bounded and continuous function. If $y_n$ is a sequence of random variables such that,  for each $n\in \mathbb N$, $y_n\sim F_n$ and $y$ is a random variable with distribution $F$, then:
	$$
	\mathbb{E}\left(g\left(y_n, a_n\right)\right) \longrightarrow \mathbb{E}\left(g(y, a)\right) .
	$$
\end{lemma}
\begin{proof}
	By Skorokhod's representation theorem there exists a probability space and random variables $Y_n, Y$ defined on it such that he distribution of $Y_n$ is $F_n$, the distribution of $Y$ is $F$ and $Y_n \rightarrow Y$ almost surely.
	
	Since \(a_n\to a\) and \(g\) is continuous on \(\mathbb{R}^2\), we have \(g(Y_n, a_n)\to g(Y, a)\) almost surely. Moreover \(g\) is bounded. Then by the dominated convergence theorem,
	\[
	\mathbb{E}\big(g(Y_n,a_n)\big)\longrightarrow \mathbb{E}\big(g(Y,a)\big).
	\]
	Note that \(\mathbb{E}(g(Y_n,a_n)) = \int g(y,a_n)\,dF_n(y) = \mathbb{E}(g(y_n,a_n))\) and \(\mathbb{E}(g(Y,a)) = \int g(y,a)\,dF(y)=\mathbb{E}(g(y,a))\). From this we obtain the desired result.
\end{proof}
\begin{lemma}\label{lema:mcont} Under Assumptions A1 to A8 and A10 the function $m$ is continuous.
\end{lemma}
\begin{proof}
	Let $h(\theta,\gamma)=\mathbb{E}_{\theta}\left(  \rho\left(  t(y,\gamma
	)\right)  \right)$, then
	$
	m(\theta)=\operatorname{argmin}_{\gamma\in\Theta}h(\theta,\gamma).$
	By Lemma \ref{lema:victor} this function is continuos in
	$(\theta,\gamma)$. 
	Suppose  that $m(\theta)$ is not continuous, then there exists a sequence
	$\theta_{n}$, $n\geq 1,$ such that $\theta_{n}\rightarrow\theta_{0}$ and
	$m(\theta_{n})\nrightarrow m(\theta_{0}).$ Then, there exists a subsequence
	\ $\theta_{n_{i}},$ $i\geq1\ $ such that $\lim_{i\rightarrow\infty}%
	\theta_{n_{i}}=\theta_{0},$ and $\lim_{i\rightarrow\infty}m(\theta_{n_{i}%
	})=\gamma$ $\neq m(\theta_{0}),$ \ Then, by the continuity of $h,$ we have
	\begin{equation}
		\lim_{i\rightarrow\infty}h(\theta_{n_{i}},m(\theta_{n_{i}}))=h(\theta
		_{0},\gamma)>h(\theta_{0},\ m(\theta_{0})).\label{con1}%
	\end{equation}
	Moreover,
	\begin{equation}
		\lim_{i\rightarrow\infty}h(\theta_{n_{i}},m(\theta_{0}))=\ h(\theta
		_{0},\ m(\theta_{0})).\label{con2}%
	\end{equation}
	Then,  (\ref{con1}) and (\ref{con2}) imply that there exists $i_{0}$ such that
	for $i>i_{0\text{ }},$ $h(\theta_{n_{i}},m(\theta_{n_{i}}))>h(\theta_{n_{i}%
	},m(\theta_{0})),$ contradicting the definition of $m(\theta_{n_{i}}).$
\end{proof}
Theorem 1 is proved as a particular case of Theorem 4 in
\cite{valdora2014robust}. Let $(y,\mathbf{x}_{1}),\dots,(y_{n},\mathbf{x}%
_{n})$ be a sample with $y_{i}\in$ $\mathbb{R},\mathbf{x}_{i}\in\mathbb{R}^{p}$, {following a GLM with regression parameter $\boldsymbol\beta_0$} and let $\Lambda:$ $\mathbb{R}\times\mathbb{R}^{p}\times\mathbb{R}%
^{p}\times\mathbb{R}^{q}\rightarrow\mathbb{R}$. This theorem considers a
general estimator of {$\boldsymbol\beta_0$} defined as
\begin{equation}
	\widehat{\boldsymbol{\beta}}_{n}=\operatorname{argmin}_{\boldsymbol{\beta}%
		\in\mathbb{R}^{p}}\sum_{i=1}^{n}\Lambda(y_{i},\mathbf{x}_{i},\boldsymbol{\beta
	},\boldsymbol{\eta}_{n}), \label{eq: def.estim.gral}
\end{equation}
where $\boldsymbol{\eta}$ is a nuisance parameter with values in
$\mathbb{R}^{q}$ and $\widehat{\boldsymbol{\eta}}_{n}$ is a sequence of
estimators of $\boldsymbol{\eta}$. Theorem 4 in
\cite{valdora2014robust}  states the strong consistency of
$\widehat{\boldsymbol{\beta}}_{n}$ to the value $\boldsymbol{\beta}_{0}$
assuming that the following properties hold:
\begin{enumerate}
	[label=\textbf{P\arabic*}, start=0]
	\item \label{prop:P0} The sequence of estimators $\widehat{\boldsymbol{\eta}%
	}_{n}$ converges almost surely to $\boldsymbol{\eta}_{0}.$
	\item \label{prop:P1} The function $\Lambda$ is continuous and bounded and
	there exists a function $\vartheta(\mathbf{x,\theta)}:$ $\mathbb{R}^{p}%
	\times\mathbb{R}^{q}\rightarrow\mathbb{R}$ and a constant $C$ such that
	\[
	\left\vert \Lambda(y,\mathbf{x},\boldsymbol{\beta},\boldsymbol{\eta}%
	_{2})-\Lambda(y,\mathbf{x},\boldsymbol{\beta},\boldsymbol{\eta}_{1}%
	)\right\vert \leq C\left\vert \vartheta(\mathbf{x},\boldsymbol{\eta}%
	_{2})-\vartheta(\mathbf{x},\boldsymbol{\eta}_{2})\right\vert
	\]
	for all $y,\mathbf{x},\boldsymbol{\beta},\boldsymbol{\eta}_{1}$ and
	$\boldsymbol{\eta}_{2}$. Moreover $\boldsymbol{\eta}_{n}\rightarrow
	\boldsymbol{\eta}_{0}$ implies
	\[
	\sup_{\mathbf{x}}\left\vert \vartheta(\mathbf{x},\boldsymbol{\eta}%
	_{n})-\vartheta(\mathbf{x},\boldsymbol{\eta}_{0})\right\vert \rightarrow0.
	\]
	\item \label{prop:P2} The value $\boldsymbol{\beta}_{0}$ satisfies
	\begin{equation}
		\mathbb{E}_{\boldsymbol{\beta}_{0}}(\Lambda(y,\mathbf{x},\boldsymbol{\beta
		}_{0},\boldsymbol{\eta}_{0}))<\mathbb{E}_{\boldsymbol{\beta}_{0}}%
		(\Lambda(y,\mathbf{x},\boldsymbol{\beta},\boldsymbol{\eta}_{0}))
	\end{equation}
	for all $\boldsymbol{\beta}\neq\boldsymbol{\beta}_{0}$.
	\item \label{prop:P3}
	There exists a function $\Lambda^{\ast}(y, \mathbf{x}, j)$, $j = -1, 0, 1$
	such that $\lim_{\gamma\rightarrow\infty}\Lambda(y, \mathbf{x}, \gamma
	\mathbf{t}, \boldsymbol{\eta}_{0})=\Lambda^{\ast}(y, \mathbf{x},
	\text{sign}(\mathbf{t}^{\prime}\mathbf{x}), \boldsymbol{\eta}_{0})$ and, if
	$\mathbf{t}^{\prime}\mathbf{x}$ $\neq0$, there exists a neighborhood of
	$\mathbf{t}$ where this convergence is uniform. Moreover
	\begin{equation}
		\tau=\inf_{||\mathbf{t||=1}}\left[  \mathbb{E}_{\boldsymbol{\beta}_{0}
		}(\boldsymbol{\Lambda}^{\ast}(y, \mathbf{x}, \text{sign}(\mathbf{t}^{\prime
		}\mathbf{x}),\boldsymbol{\eta}_{0})- \mathbb{E}_{\boldsymbol{\beta}_{0}
		}(\Lambda(y,\mathbf{x}, \boldsymbol{\beta}_{0}, \boldsymbol{\eta}_{0})\right]
		>0.\label{eq:tau}%
	\end{equation}
	\item \label{prop:P4} $P(\mathbf{t}^{\prime}\mathbf{x}=0)<\tau/M$ for all
	$\mathbf{t}$ with $\Vert\mathbf{t}\Vert=1$ where $M=\sup_{y, \mathbf{x}, \boldsymbol{\beta}}\Lambda(y, \mathbf{x}, \boldsymbol{\beta}, \boldsymbol{\eta
	}_{0})$, and $\tau$ is defined in \eqref{eq:tau}.
\end{enumerate}
\begin{proof}
	[Proof of Theorem 1] 
	Note that the MNQPIT estimator is given by (\ref{eq: def.estim.gral} ) with
	$\boldsymbol{\eta}=\left(  \boldsymbol{\mu},\boldsymbol{\Sigma}\right)  $,
	\[
	\Lambda\left(  y,\mathbf{x},\boldsymbol{\beta},\boldsymbol{\mu}%
	,\boldsymbol{\Sigma}\right)  =\rho\left(  t(y,m\left(  g^{-1}\left(
	\mathbf{x}^{\top}\boldsymbol{\beta}\right)  \right)  \right)  w\left(
	\mathbf{x},\boldsymbol{\mu},\boldsymbol{\Sigma}\right)
	\]
	and $\vartheta(\mathbf{x},\boldsymbol{\eta})=w(\mathbf{x},\boldsymbol{\mu
	},\boldsymbol{\Sigma})$. Then it will be enough to show properties \ref{prop:P0}-\ref{prop:P4}. {Note that, by Assumptions A8 and A12, $M=1$}.
	
	\ref{prop:P0} follows from A11 and \ref{prop:P1} follows from
	Assumption A12 and Lemma 4 in \cite{valdora2014robust}. To prove
	\ref{prop:P2}, first recall that, by A5, if $y\sim
	F(.,\theta_{0})$ and $\theta\neq\theta_{0}$,
	\begin{equation}
		\mathbb{E}_{\theta_{0}}\left(  \rho(t(y,m(\theta)))\right)  -\mathbb{E}%
		_{\theta_{0}}\left(  \rho(t(y,m(\theta_{0})))\right)  >0. \label{eq:diffD}%
	\end{equation}
	Now take $\boldsymbol{\beta}\neq\boldsymbol{\beta}_{0}$ and let $V=\left\{
	\mathbf{x}:\mathbf{x}^{\top}\left(  \boldsymbol{\beta}-\boldsymbol{\beta}%
	_{0}\right)  \neq\mathbf{0}\right\}  \cap\left\{  \mathbf{x}:w\left(
	\mathbf{x},\boldsymbol{\mu}_{0},\boldsymbol{\Sigma}_{0}\right)  >0\right\}  $,
	then by A13 and \eqref{eq:diffD}, $\mathbb{P}(V)>0$ and then, for
	$\mathbf{x}\in V$,
	\begin{equation}
		\mathbb{E}_{\boldsymbol{\beta}_{0}}\left(  \Lambda\left(  y,\mathbf{x}%
		,\boldsymbol{\beta},\boldsymbol{\mu}_{0},\boldsymbol{\Sigma}_{0}\right)
		-\Lambda\left(  y,\mathbf{x},\boldsymbol{\beta}_{0},\boldsymbol{\mu}%
		_{0},\boldsymbol{\Sigma}_{0}\right)  \mid\mathbf{x}\right)  >0.
		\label{eq:difDw}%
	\end{equation}
	From \eqref{eq:difDw},
	\[
	\begin{aligned}
		& \mathbb E_{\boldsymbol{\beta}_0}\left(\Lambda\left(y, \mathbf{x}, \boldsymbol{\beta}, \boldsymbol{\mu}_0, \boldsymbol{\Sigma}_0\right)\right)- \mathbb E_{\boldsymbol{\beta}_0}\left(\Lambda\left(y, \mathbf{x}, \boldsymbol{\beta}_0, \boldsymbol{\mu}_0, \boldsymbol{\Sigma}_0\right)\right) \\
		& \left.\quad \geq   \mathbb E_{
			\boldsymbol{\beta}_0}\left(\Lambda\left(y, \mathbf{x}, \boldsymbol{\beta}, \boldsymbol{\mu}_0, \boldsymbol{\Sigma}_0\right)-\Lambda\left(y, \mathbf{x}, \boldsymbol{\beta}_0, \boldsymbol{\mu}_0, \boldsymbol{\Sigma}_0\right) \right)
		I(V)\right) >0,
	\end{aligned}
	\]
	which implies \ref{prop:P2}.
	
	To prove \ref{prop:P3}, note that it is immediate that 
	$\boldsymbol{\Lambda}^{\star}$ exists and is given by%
	\[
	\Lambda^{\ast}(y\mathbf{,x,}j,\boldsymbol{\eta})=\left\{
	\begin{array}
		[c]{lll}%
		\lim_{\theta\rightarrow\theta_{1}}\rho\left(  t(y,m(\theta)\right)  w\left(
		\mathbf{x},\boldsymbol{\mu},\boldsymbol{\Sigma}\right)  & \text{if} & j=-1\\
		\rho\left(  t(y,m(1)\right)  w\left(  \mathbf{x},\boldsymbol{\mu
		},\boldsymbol{\Sigma}\right)  & \text{if} & j=0\\
		\lim_{\theta\rightarrow\theta_{2}}\rho\left(  t(y,m\left(  \theta\right)
		\right)  w\left(  \mathbf{x},\boldsymbol{\mu},\boldsymbol{\Sigma}\right)  &
		\text{if} & j=1.
	\end{array}
	\right.
	\]
	Then, it is enough to show that $\tau>0$, where $\tau$ is defined in
	\eqref{eq:tau}. {Using A13 and properties of the infimum, it can be
		seen that there exists $\zeta>0$ and $\delta>0$ such that
		\begin{equation}
			\inf_{\mathbf{t}\in\mathbf{S}}\mathbb{P}\left(  \left\{  \mathbf{x}^{\top
			}\mathbf{t}\neq\mathbf{0}\right\}  \cap\left\{  \omega\left(  \mathbf{x}%
			,\boldsymbol{\mu}_{0},\Sigma_{0}\right)  >\zeta\right\}  \right)  \geq
			\delta/2. \label{eq:infendem1}%
		\end{equation}
		We can also find $K_{1}$ and $K_{2}$ such that
		\begin{equation}
			\mathbb{P}\left(  \mathbf{x}^{\top}\boldsymbol{\beta}\in\left[  K_{1}%
			,K_{2}\right]  \right)  >1-\delta/4. \label{eq:Pendem1}%
		\end{equation}
		Let $\left.  V_{\mathbf{t}}=\left\{  \mathbf{x}^{\top}\mathbf{t}\neq
		\mathbf{0}\right\}  \cap\left\{  \omega\left(  \mathbf{x},\boldsymbol{\mu}%
		_{0},\Sigma_{0}\right)  >\zeta\right\}  \cap\left\{  \mathbf{x}^{\top
		}\boldsymbol{\beta}\in\left[  K_{1},K_{2}\right]  \right\}  \right]$,
		by \eqref{eq:infendem1} and \eqref{eq:Pendem1} we have $\mathbb{P}\left(  V_{t}\right)  >\delta/4$
		for all $\mathbf{t}\in S$. Using A5, it can be seen that
		for all $\mathbf{x}\in\mathbb{R}^{p}$ and $\mathbf{t}\in S$
		\[
		\mathbb{E}_{\boldsymbol{\beta}_{0}}\left(  \left(  \boldsymbol{\Lambda}^{\ast
		}\left(  y,\mathbf{x},\operatorname{sign}\left(  \mathbf{x}^{\top}%
		\mathbf{t}\right)  \right)  -\Lambda\left(  y,\mathbf{x},\boldsymbol{\beta
		},\boldsymbol{\mu}_{0},\boldsymbol{\Sigma}_{0}\right)  \right)  \mid
		\mathbf{x}\right)  \geq0.
		\]
		Let
		$C_{i}(\theta)=\mathbb{E}_{\theta_{0}}\left(  \rho(t(y,m_{i}))\right)
		-\mathbb{E}_{\theta}\left(  \rho(t(y,m(\theta)))\right)  .$ Then for all
		$\theta\in\left[  g\left(  K_{1}\right)  ,g\left(  K_{2}\right)  \right]  $
		and $i=1,2$, we have that $C_{i}(\theta)$ is positive and continuous. Then
		\[
		c_{0}=\min\left\{  \min_{\theta\in\left[  g\left(  K_{1}\right)  ,g\left(
			K_{2}\right)  \right]  }C_{1}(\theta),\min_{\theta\in\left[  g\left(
			K_{1}\right)  ,g\left(  K_{2}\right)  \right]  }C_{2}(\theta)>0\right\}  >0
		\]
		and
		\[
		\begin{aligned}
			\mathbb E_{\beta_0} & \left(\Phi^*\left(y, \mathbf{x}, \operatorname{sign}\left(\mathbf{t}^{\prime} \mathbf{x}\right)\right)-\Lambda\left(y, \mathbf{x}, \boldsymbol{\beta}_0^{\prime} \mathbf{x}, \boldsymbol{\mu}_0, \boldsymbol{\Sigma}_0\right)\right) \\
			& \geq \mathbb E\left(\mathbb E_{\boldsymbol{\beta}_0}\left(\left(\Phi^*\left(y, \mathbf{x}, \operatorname{sign}\left(\mathbf{t}^{\prime} \mathbf{x}\right)\right)-\Lambda\left(y, \mathbf{x}, \boldsymbol{\beta}^{\prime} \mathbf{x}\right)\right) \mid \mathbf{x}\right) \mathbf{I}\left(\mathbf{x} \in V_{\mathbf{t}}\right)\right) \\
			& \geq c_0 \mathbb E\left(w\left(\mathbf{x}, \boldsymbol{\mu}_0, \boldsymbol{\Sigma}_0\right) \mathbf{I}\left(\mathbf{x} \in V_{\mathbf{t}}\right)\right) \\
			& \geq c_0 \zeta \delta / 4 .
		\end{aligned}
		\]
	} This implies that $\tau\geq c_{0}\zeta\delta/4$ and therefore \ref{prop:P3}
	holds. Finally, note that \ref{prop:P4} is one of the assumptions of the theorem.
\end{proof}
\subsection{Asymptotic normality}\label{sec:Sasnorm}
\begin{lemma}\label{lema:mderviable} If $y$ is discrete or absolutely continuous and Assumptions A5 and A15 to A17 hold, then
	$m$ is twice differentiable.
\end{lemma}
\begin{proof} Recall the definition of $m$ given in equation (6) and note that $t$ is differentiable by A15 and A16. Then
	\begin{equation*}
		\mathbb E_\theta\left(\psi \left(t(y,m(\theta))\right)t^{\prime}(y,m(\theta))\right) =0.
	\end{equation*}
	Then, because of the Implicit Function Theorem and Assumption A17, $m$ is differentiable and 
	$$m^{\prime}(\theta)= - \displaystyle\frac{\mathbb E_\theta\left(\psi\left( t(y,m(\theta))\right)t^{\prime}(y,m(\theta)) f^{\prime}(y,\theta)/f(y,\theta) \right)}{\mathbb E_\theta\left(\psi^{\prime}\left(t(y,m(\theta))\right)\left(t^{\prime}(y,m(\theta))\right)^2+\psi \left(t(y,m(\theta))\right)t^{\prime\prime}(y,m(\theta))\right)},$$ where $f(y,\theta)$ is the probability density function (if $y$ is absolutely continuous) or the probability mass function (if $y$ is discrete).
	Using again Assumptions A15, A16 and A17, we conclude that $m^{\prime}$ is also differentiable.
\end{proof}
\begin{proof}[Proof of Theorem 2]
	The proof is analogous to the proof of Theorem 2 in \cite{valdora2014robust} with $\Phi\left(y, \mathbf{x}, \boldsymbol{\beta}, \boldsymbol{\mu}, \boldsymbol{\Sigma}\right) =\rho\left(t(y, mg^{-1}\left(\mathbf{x}^{\top} \boldsymbol{\beta}  \right)\right) w\left(\mathbf{x}, \boldsymbol{\mu}, \boldsymbol{\Sigma}\right)$.
\end{proof}
\begin{proof}[Proof of Theorem 3]
	The proof follows the same lines as Theorem 3 in  \cite{valdora2014robust}.
	
	Suppose that $\varepsilon$ is a positive real number such that there exists a
	sequence of distribution functions $H_{k}$ such that $\left\| T((1-\varepsilon
	)H_{0}+\varepsilon H_{k})\right\|_2\rightarrow\infty$ as $k\rightarrow\infty$, where$||$ $||_{2}$ denotes the $l_{2}$ norm.
	
	Let $\boldsymbol{\beta}_{k}=T((1-\varepsilon)H_{0}+\varepsilon H_{{k}})$, $\theta_{\mathbf x}^k=g^{-1}\left(\mathbf{x}^{\top} \boldsymbol{\beta}_{k}\right)$ and $\theta_{\mathbf x}^0=g^{-1}\left(\mathbf{x}^{\top} \boldsymbol{\beta}_{0}\right)$ . Then, by A7, A8 and the definition of $\mathbf T$, 
	\begin{align}
		&  (1-\varepsilon)\mathbb E_{H_{0}}\left(\rho
		\left(t\left(y,m\left(\theta_{\mathbf x}^k\right)\right)\right)\right)\nonumber\\
		&  \leq(1-\varepsilon)\mathbb E_{H_{0}}\left(\rho
		\left(t\left(y,m\left(\theta_{\mathbf x}^k\right)\right)\right)\right)+\varepsilon \mathbb E_{H_{k}}\left(\rho
		\left(t\left(y,m\left(\theta_{\mathbf x}^k\right)\right)\right)\right)\nonumber\\
		& \nonumber   \leq(1-\varepsilon)\mathbb E_{H_{0}}\left(\rho
		\left(t\left(y,m\left(\theta_{\mathbf x}^0\right)\right)\right)\right)+\varepsilon \mathbb E_{H_{k}}\left(\rho
		\left(t\left(y,m\left(\theta_{\mathbf x}^k\right)\right)\right)\right)\label{ineq1}\\
		&  \leq(1-\varepsilon)\mathbb E_{H_{0}}\left(\rho
		\left(t\left(y,m\left(\theta_{\mathbf x}^0\right)\right)\right)\right)+\varepsilon.
	\end{align}
	\ Let $\boldsymbol\alpha_{k}= \boldsymbol{\beta}_{k}/||\boldsymbol{\beta
	}_{k}||$,  then by taking a subsequence, we may assume without loss of generality that $\boldsymbol \alpha
	_{k}\rightarrow\alpha$. Then, using  A5, A8, and Lemma \ref{lem:limt}, 
	\begin{align}
		\nonumber	&  \lim_{k\rightarrow\infty}(1-\varepsilon)\mathbb E_{H_{0}}\left(\rho
		\left(t\left(y,m\left(\theta_{\mathbf x}^k\right)\right)\right)\right)\\
		&  \geq  (1-\varepsilon) \mathbb E_{H_{0}}\left(  \phi_1^*(y ) I(\mathbf{x}^\top\boldsymbol{\alpha}<0)+ \phi_2^*(y )I(\mathbf{x}^\top\boldsymbol{\alpha}>0)\right)
		\nonumber\\
		&  \geq(1-\varepsilon) \mathbb E_{H_{0}}\left(\min\left(\phi_1^*(y), \phi_2^*(y) \right)\right). \label{ineq2}%
	\end{align}
	\bigskip Combining inequalities (\ref{ineq1}) and (\ref{ineq2}) we obtain
	\[
	(1-\varepsilon) \mathbb E_{H_{0}}\left(\min\left(\phi_1^*(y), \phi_2^*(y) \right)\right) \leq(1-\varepsilon)\mathbb E_{H_{0}}\left(\rho
	\left(t\left(y,m\left(\theta_{\mathbf x} \right)\right)\right)\right)+\varepsilon,
	\] and the result follows.
\end{proof}
\subsection{Breakdown point}\label{sec:SBP}
\begin{lemma}\label{lemma:m1}  Assume A1 to A3, A5, A7 to A10, A19 and A20. Then 
	$\lim_{\theta\rightarrow\theta_1} m(\theta)=\theta_1.$
\end{lemma}
\begin{proof} Let $R(\theta, \lambda) = \mathbb E_\theta \left( \rho\left( t(y, \lambda) \right) \right)$, $\lambda \in \Theta$. By Assumptions A1,  A3, A8, A10, A19 and A20, 
	$$R(\theta, \lambda) =\sum_{y\in R_y}\rho\left( t(y, \lambda) \right)  p(y,\theta)\stackrel{\theta\rightarrow \theta_1}{\longrightarrow}\rho\left( t(0, \lambda) \right)=\rho\left( \Phi^{-1}\left( \frac{p ( 0, \lambda)}{2} \right)\right).$$ 
	Suppose  $m(\theta)$ converges to $m_{0}\in (0, +\infty)$ when $\theta\rightarrow \theta_1$. Let  $m_3 \in (\theta_1, m_0)$, then, by A2, A7
	and A9.
	$$0<p(0, m_0)<p(0, m_3)<1 \Rightarrow \Phi^{-1}\left(\frac{p(0, m_0)}{2}\right)< \Phi^{-1}\left(\frac{p(0, m_3)}{2}\right)<0,$$  
	$$ \Rightarrow\rho\left(  \Phi^{-1}\left(\frac{p(0, m_3)}{2}\right) \right)<\rho\left(  \Phi^{-1}\left(\frac{p(0, m_0)}{2}\right)\right).$$
	This implies that $$\lim_{\theta\rightarrow \theta_1 } R(\theta, m_3) < \lim_{\theta\rightarrow \theta_1} R(\theta, m_0),$$  and this means that there exists $\delta>0$ such that, if $\theta<\theta_1 + \delta$ then $ R(\theta, m_3)  <  R(\theta, m_0)$ and therefore $\operatorname{argmin}_{\lambda \in (\theta_1,\theta_2)} R(\theta, \lambda) \neq m_0$. This is a contradiction.
\end{proof}
\begin{lemma} \label{lemma:m2}
	Assume  A1 to A6, A8 and A10. Then
	$\lim_{\theta\rightarrow \theta_2}m(\theta)=\theta_2$.
\end{lemma}
\begin{proof}
	We will show that given any $\eta_{0}\in(\theta_1, \theta_2)$, there exists $\theta_{0} \in(\theta_1, \theta_2)$ such that $m(\theta)>\eta_{0}$ for all
	$\theta\in(\theta_{0}, \theta_2).$ Call \
	\begin{equation}
		R(\theta,\eta)=\mathbb E_{\theta}\left(  \rho(\Phi^{-1}(F(y,{\eta
		})-0.5 p(y,{\eta}))\right) , \label{1}%
	\end{equation} 
	where $y\sim F(., \theta)$. Recall that
	\[m(\theta)=\operatorname{argmin}_{\eta \in (\theta_1, \theta_2)} R(\eta,\theta)
	\]
	and that, by Lemma \ref{lemma:epsilon0}, there exists $\epsilon_0>0$ such that $R(\theta, \theta)<1-\epsilon_0$.
	
	Then, it it will be enough to show that given
	${\eta}_{0}$ there exists $\theta_{0}$ such that for
	${\eta}\leq{\eta}_{0}$ and $\theta\geq\theta_{0}$,
	$R($${\eta}$$,\theta)> 1 - \epsilon_0.$
	
	Since $\ \lim_{z\rightarrow1}\Phi^{-1}(z)=\infty,$ and $\lim_{z\rightarrow
		\infty}\rho(z)=1,$ there exists $\delta_{0}$ such that%
	\begin{equation}
		(1-\delta_{0})\rho(\Phi^{-1}(1-\delta_{0}))> 1- \epsilon_0. \label{2}%
	\end{equation}
	Given $\eta_{0},$ we can find $y_{0}$ such that
	\begin{equation}
		F(y,\eta)>1-\delta_{0}\;\;\forall\eta\leq\eta_{0},\forall y\geq y_{0}
		\label{3}%
	\end{equation}
	and we can also find $\theta_{0}$ such that
	\begin{equation}
		F(y,\theta)\leq\delta_0,\text{ }\forall\theta\geq\theta_{0},\forall y\leq
		y_{0}.\label{5}%
	\end{equation}
	Then, using (\ref{2}) (\ref{3}) and (\ref{5}) we get
	that for any $\eta\leq\eta_{0},\theta$ and $\theta\geq\theta_{0}$ \ \
	\begin{align*}
		R(\theta,\eta) &  \geq \mathbb E_{\theta}(\rho(\Phi^{-1}(F(y,{\eta
		})-0.5 p(y,{\eta}))I(y  > y_{0} ))\\
		&  \geq \mathbb E_{\theta}(\rho(\Phi^{-1}(F(y-1, {\eta
		}))I(y  > y_{0} ))\\
		&  \geq \mathbb E_{\theta}(\rho(\Phi^{-1}(1-\delta_{0}))I(y{>
			y}_{0}))\\
		&  \geq\rho(\Phi^{-1}(1-\delta_{0}))(1-\delta_{0}),\ \\
		&  > 1- \epsilon_0.\
	\end{align*}
	Then since $R(\theta,\theta)<1 - \epsilon_0,$ this implies that $m(\theta)\geq\eta_{0}$
	for all $\theta\geq\theta_{0}$ and therefore $\ \lim_{\theta\rightarrow\theta_2
	}m(\theta)=\theta_2.$
\end{proof}
\begin{proof} [Proof of Theorem 4]
	Note that $t_0(0, \theta) = p(0, \theta)/2$ while, by Assumption A19, $t_0(y, \theta) =  p(0, \theta) + \sum_{k=1}^y p(k, \theta) -   p(y, \theta)/2$  for $y>0$.
	Then
	$\lim_{\theta\rightarrow \theta_1} t_0(y, \theta) = 1/2 I_{\left\{0\right\}}(y) + I_{\left(0, +\infty\right)}(y)$.  
	Recall that $t = \Phi^{-1} \circ t_0$, $\Phi^{-1}(1/2)=0$ and $\lim_{x\rightarrow 1}\Phi^{-1}(x) = +\infty$, then, by Assumptions A8, A10, A20 and Lemma \ref{lemma:m1}, $\phi^*_1(y)=\lim_{\theta\rightarrow \theta_1} \rho\left( t(y, m(\theta))\right)= \lim_{\theta\rightarrow \theta_1} \rho \left(t(y, \theta)\right) =  I_{\left(0, +\infty\right)}(y)$.  
	
	We now compute $\phi^*_2(y)$. By A3, $\lim_{\theta \rightarrow \theta_2} p(y,\theta)=0$ for all $y$. Then, by the properties of $\Phi^{-1}$,   A8, A10 and Lemma \ref{lemma:m2}, 
	$\phi^*_2(y)=\lim_{\theta\rightarrow \theta_2} \rho \left( t(y, m(\theta)) \right)= \lim_{\theta\rightarrow \theta_2} \rho \left( t(y, \theta) \right)=  1.$
	Then, by Theorem 3, we get that the breakdown point is at least
	$$\varepsilon_0=\frac{\mathbb E_{H_0}\left( I_{(0, +\infty)} (y) \right) - \mathbb E_{H_0}\left(\rho \left( t \left(y, m(\theta_{\mathbf{x}})  \right)\right)\right)}{ 1 +\mathbb E_{H_0}\left( I_{(0, +\infty)} (y) \right)-\mathbb E_{H_0}\left(\rho \left(t \left(y, m(\theta_{\mathbf{x}})\right)\right)\right)},
	$$
	concluding the proof.
\end{proof}


\bibliography{mireg}

\end{document}